\documentclass[a4paper]{quantumarticle}
\pdfoutput=1

\usepackage{algorithm}
\usepackage{algpseudocode}

\usepackage{graphicx}
\usepackage{amsmath,amssymb,amsfonts}
\usepackage{amsthm}
\newtheorem{proposition}{Proposition}
\theoremstyle{remark}
\newtheorem{remark}{Remark}
\theoremstyle{plain}
\usepackage[english]{babel}
\usepackage{braket}
\usepackage{booktabs}
\usepackage{siunitx}
\usepackage[version=4]{mhchem}
\AtBeginDocument{\RenewCommandCopy\qty\SI}

\usepackage[colorlinks=true,hyperindex,breaklinks=true,allcolors=quantumviolet]{hyperref}

\begin{document}

\title{Logarithmic depth compression of Heisenberg Hamiltonian simulation by fan-out parallelization, with built-in error detection}
\author{Artemiy Burov}
\affiliation{School of Life Sciences, University of Applied Sciences Northwestern Switzerland (FHNW), Hofackerstrasse 30, CH-4132 Muttenz, Switzerland}
\affiliation{\'{E}cole Polytechnique F\'{e}d\'{e}rale de Lausanne (EPFL), CH-1015 Lausanne, Switzerland}
\author{Cl\'{e}ment Javerzac}
\email{clement.javerzac@fhnw.ch}
\affiliation{School of Life Sciences, University of Applied Sciences Northwestern Switzerland (FHNW), Hofackerstrasse 30, CH-4132 Muttenz, Switzerland}
\affiliation{\'{E}cole Polytechnique F\'{e}d\'{e}rale de Lausanne (EPFL), CH-1015 Lausanne, Switzerland}
\affiliation{MatterDecoder, CH-3008 Bern, Switzerland}

\begin{abstract}
Noisy intermediate-scale quantum computers are constrained by circuit depth, while many Hamiltonian simulation methods, such as product-formula simulation of spin systems, lead to narrow and deep circuits. Here we introduce a fan-out-based gadget compiler that trades circuit depth for width in simulations of Heisenberg-type nuclear magnetic resonance (NMR) Hamiltonians. Each logical spin is encoded into a small repetition-code register sized by its interaction degree, so that all pairwise interactions of a given Pauli type execute in parallel after a logarithmic-depth CNOT fan-out, and the redundant registers provide error detection for post-selection at no additional algorithmic overhead. The central result is a resource comparison of the two compilations under a fixed protocol, transpiled to a heavy-hex superconducting coupling map and to all-to-all trapped-ion connectivity across a set of NMR spin systems. For interaction graphs with a high-degree hub the volume-optimal schedule reduces the two-qubit depth by a factor of two and the volume by a factor of 1.7 for the 13-spin demonstration, which on heavy-hex also carries a lower two-qubit gate count, and the depth reduction rises to 2.5-fold on all-to-all connectivity for the highest-degree molecule studied. On all-to-all the two-qubit gate count rises for every system, so the volume reduction is a benefit on depth-limited hardware. The gain grows with the degree inhomogeneity of the interaction graph and vanishes for dense uniform graphs, where the schedule family falls back to the sequential circuit. As a worked example we simulate the zero-field NMR spectrum of tetramethylsilane, a 13-spin star system. Under a noise model scaled from the published calibration of a present-day quantum processor, the shallower gadget circuits match or surpass the sequential compilation only after post-selection on their built-in error detection, once error rates improve by one to one and a half orders of magnitude. We verify the spectra against an independent classical computation.
\end{abstract}

\maketitle

\section{Introduction}
In the noisy intermediate-scale quantum (NISQ) era \cite{preskill_2018} we have strict requirements on the programs that quantum computers are able to run (Fig.~\ref{fig:milestones}). Both currently dominant architectures of quantum computers, the superconducting architecture and the trapped-ion architecture, operate within a tight qubit and gate layer budget \cite{google_2019, ustc_2021, ibm_2023, google_2024, rigetti_2024, ustc_2025, ibm_2026, ionq_2019, honeywell_2021, quantinuum_2023, ionq_2024, quantinuum_2024, quantinuum_2025, ionq_2026}. This motivates flexibility in quantum algorithms to fit the available resources. Product-formula circuits for nuclear magnetic resonance (NMR) simulation illustrate the tension: they are deep and narrow, while current quantum computers favor shallower and wider circuits. In this work we show a way to convert narrow deep circuits into wide shallow circuits for the simulation of spin Hamiltonians, and we quantify the conversion with a fixed resource-comparison protocol across a set of NMR spin systems and two hardware architectures. In previous work we approached the same tension from the complementary side, executing deep sequential NMR circuits on present hardware with error suppression \cite{burov2025largecircuitexecutionnmr}; here we reshape the circuits themselves. As a demonstration we simulate the zero-field NMR spectrum of tetramethylsilane, a 13-spin system whose star-shaped interaction graph is the geometry in which the compression is largest. Zero- to ultralow-field NMR \cite{ledbetter_2011, blanchard_2016} offers spectra with simple exact structure against which the simulation can be judged, and our reference spectra are verified with an independent classical computation \cite{SPINACH_HOGBEN}. A similar approach extends to solid-state NMR crystallography, where the problem size can scale into the classically intractable regime. Refs.~\cite{seetharam_2023, elenewski2024} study the simulation of different NMR experiments on quantum computers. Together with the ability to execute deep quantum circuits with reduced noise through error mitigation and error suppression \cite{cai_2023, burov2025largecircuitexecutionnmr}, logarithmic depth compression reduces the hardware cost of quantum utility for computational NMR \cite{burov2024quantumutilitynmrquantum}.

The mechanism of copying qubits to parallelize commuting operations dates back to Moore and Nilsson \cite{doi:10.1137/S0097539799355053} and underlies several recent constructions, from ancilla-parallelized simulation algorithms \cite{zhang_2024, boyd_2023} to parallelized fault-tolerant compilation \cite{litinski_2019, ppr_2026}. To the best of our knowledge the specific combination presented here is new for product-formula simulation of spin Hamiltonians, whose interaction terms decompose into two-local Pauli blocks of uniform type: the fan-out registers are sized by the interaction degree, a single parameter moves the schedule from the sequential circuit to the fully parallel one, and the construction is compiled and transpiled end to end to routed hardware; the extension to generic Pauli-sum Hamiltonians is left to future work. Our quantitative claim is a statistically controlled account of when this parallelization pays on hardware resources: with seed-controlled statistics on transpiled circuits we map, across a set of spin systems and two hardware architectures, where the two-qubit depth, gate count, and width-depth volume improve. The best resource tradeoff is reached at a partial fan-out, between the sequential circuit and the fully parallel one, and the gain grows with the degree inhomogeneity of the interaction graph, appearing only once a hub is sufficiently high-degree. The redundancy of the registers simultaneously provides error-detecting parity checks at no additional gate or measurement cost, recovering inside product-formula circuits the intrinsic-redundancy protection known from parity-encoded architectures for combinatorial optimization \cite{lechner_2015, pastawski_2016}. Zero-field NMR spectra have been computed on quantum hardware for a four-spin system \cite{seetharam_2023}, and larger high-field systems were executed in our previous work \cite{burov2025largecircuitexecutionnmr}; the demonstration presented here makes no scale claim: its object is the comparison of compilation strategies under a device-noise model extrapolated to lower error rates.

\begin{figure*}[tp]
    \centering
    \includegraphics[width=\textwidth]{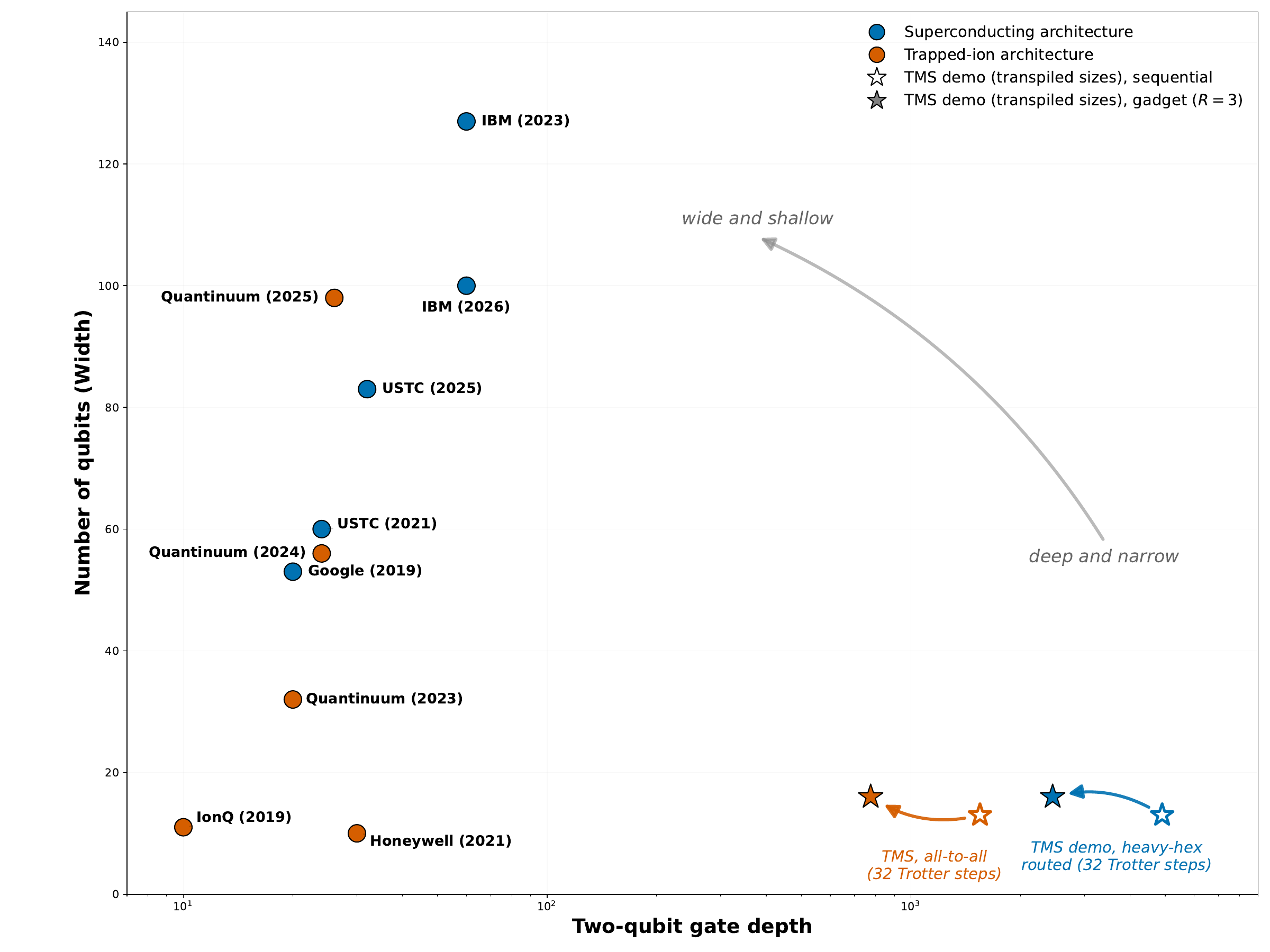}
    \caption{Circuit width versus two-qubit gate depth for published milestone experiments on quantum devices, together with the transpiled circuit sizes of this work's demonstration (stars): TMS ($^{29}$Si(CH$_3$)$_4$, 13 spins) at zero field, second-order product formula with $M=32$ Trotter steps. The stars are compiled-and-transpiled circuit sizes, computed classically rather than run on hardware. Device coordinates are the qubit counts and two-qubit-gate depths of the cited experiments (random-circuit-sampling cycles and layers of CNOT or CZ gates for superconducting devices; parallel two-qubit-gate layers for Quantinuum; for IonQ the two-qubit gate count, which for the plotted Bernstein--Vazirani circuit equals the depth since every two-qubit gate addresses the same ancilla). Arrows connect the sequential compilation (open star) to the gadget compilation at the optimal schedule $R=3$ (filled star): blue, routed to a fixed 16-qubit subgraph of ibm\_aachen (sequential 13 qubits, depth 4903; gadget 16 qubits, depth 2452); orange, all-to-all connectivity with native M{\o}lmer--S{\o}rensen gates (sequential depth 1548; gadget 775). The heavy-hex circuits are the ones used in the noisy simulation of Sec.~\hyperref[sec:error_detection]{Error detection and post-selection}, under an extrapolated-calibration noise model: the per-qubit ibm\_aachen calibration scaled to median two-qubit errors between the present-day $1.5\times10^{-3}$ and $5\times10^{-5}$, with correspondingly scaled coherence times. The demonstration depths count parallel two-qubit-gate layers, as for the superconducting and Quantinuum device points; where two-qubit gates do not run in parallel, the two-qubit gate count is the binding cost, and the all-to-all fan-out raises it (from 1548 to 2194 here). The milestone experiments were executed at heterogeneous output fidelities (from 0.03\% cross-entropy fidelities for the deepest random-circuit-sampling experiments to 78\% average algorithm success rates), so the device points indicate demonstrated circuit sizes rather than a single capability frontier; the demonstration circuits, one to two orders of magnitude deeper than the device points, are shown for context.
    \textit{Superconducting architecture:}
    Google 2019 \cite{google_2019},
    USTC 2021 \cite{ustc_2021},
    IBM 2023 \cite{ibm_2023},
    USTC 2025 \cite{ustc_2025},
    and IBM 2026 \cite{ibm_2026}.
    \textit{Trapped-ion architecture:}
    IonQ 2019 \cite{ionq_2019},
    Honeywell 2021 \cite{honeywell_2021},
    Quantinuum 2023 \cite{quantinuum_2023},
    Quantinuum 2024 \cite{quantinuum_2024},
    and Quantinuum 2025 \cite{quantinuum_2025}.
    }
    \label{fig:milestones}
\end{figure*}

We map the liquid-state NMR Hamiltonian of $N$ nuclear spins \cite{levitt, burov2024quantumutilitynmrquantum} onto the following qubit Hamiltonian:
\begin{equation}
    H = \sum_{k=1}^N \omega_{k} I_k^X + 2\pi \sum_{k<l} J_{kl} \vec{I}_k \cdot \vec{I}_l,
    \label{NMRH2}
\end{equation}
where $\omega_k$ is the offset frequency of spin $k$ in the rotating frame (determined by its chemical shift $\delta_k$, the static field strength $B_0$, and the gyromagnetic ratio $\gamma$ of the nucleus), $J_{kl}$ is the scalar coupling constant between spins $k$ and $l$, and $\vec{I}_k = (I_k^X, I_k^Y, I_k^Z)$ is the spin-$\tfrac{1}{2}$ angular momentum operator. The Zeeman term is written along $X$ instead of the conventional $Z$: a global relabeling of the spin axes that leaves the spectrum unchanged and places the detected magnetization on the qubit measurement axis $Z$. Zero-field experiments, including the demonstration in Sec.~\hyperref[sec:error_detection]{Error detection and post-selection}, correspond to the special case $\omega_k = 0$, in which only the coupling terms remain; the parallelization introduced below acts on the coupling terms and therefore applies unchanged in both regimes.
Here we propose a method for logarithmic compression of the depth of the product formula circuits required for the simulation of the Hamiltonian in Eq.~\ref{NMRH2}, at the cost of an increase in the number of required qubits from $N$ to $\sum_k \max(1, d_k)$, where $d_k$ is the degree of spin $k$ in the interaction graph. This overhead is quadratic for fully connected systems and linear for sparse topologies such as chains. As with similar depth compression methods described in other contexts \cite{loke2026distributedquantumcomputingdistributed, gokhale2020quantumfanoutcircuitoptimizations}, compression is achieved by using the fact that the Z-basis repetition code has a logical Z operator of Hamming weight $d_Z=1$. This allows us to parallelize interactions that address the same qubit by copying that qubit via a binary CNOT tree. The depth of the binary tree scales logarithmically with the number of parallelized interactions. The qubits used for copying are freed upon executing the interactions and can be used for error detection and post-selection.

The paper is organized as follows. We first write the product-formula circuits for the liquid-state NMR Hamiltonian, then introduce the fan-out gadget and its rounds-parameterized schedule and report the resource comparison across spin systems and two hardware architectures (Sec.~\hyperref[sec:gadget]{Parallelization/gadget}). Section~\hyperref[sec:error_detection]{Error detection and post-selection} develops the built-in error detection and demonstrates the full pipeline on the zero-field spectrum of tetramethylsilane under a device-noise model extrapolated to lower error rates. We close with an outlook toward fermionic and central-spin applications; the appendices collect the noise model, compilation and transpilation statistics, and the connectivity and star-scaling studies.

\begin{figure}[htbp]
    \centering
    \includegraphics[width=\textwidth/3]{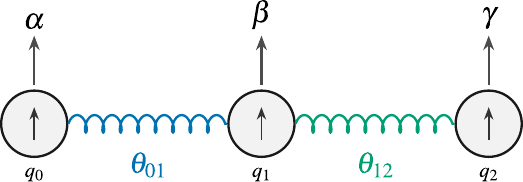}
    \caption{Example of a $3$-spin molecule with chain connectivity mapped to qubits $\{ q_0, q_1, q_2 \}$ with respective Zeeman frequencies $\{ \alpha, \beta, \gamma \}$ and spin-spin coupling parameters $\theta_{01}$ and $\theta_{12}$.}
    \label{fig:chain}
\end{figure}

\section{Product formula circuits for liquid-state NMR without secular approximation}

Eq.~\ref{pf_exp} shows the first-order, one-step product formula \cite{lloyd_1996, suzuki_1991} for the Hamiltonian in Eq.~\ref{NMRH2} without secular approximation. Higher-order formulas and multiple Trotter steps can be used to improve accuracy \cite{childs_2021} at the cost of deeper circuits; the parallelization introduced in Sec.~\hyperref[sec:gadget]{Parallelization/gadget} applies identically regardless of the chosen order or step count, since it acts on each block of same-type Pauli terms independently.

\begin{figure*}[tp]
    \centering
    \includegraphics[width=\textwidth]{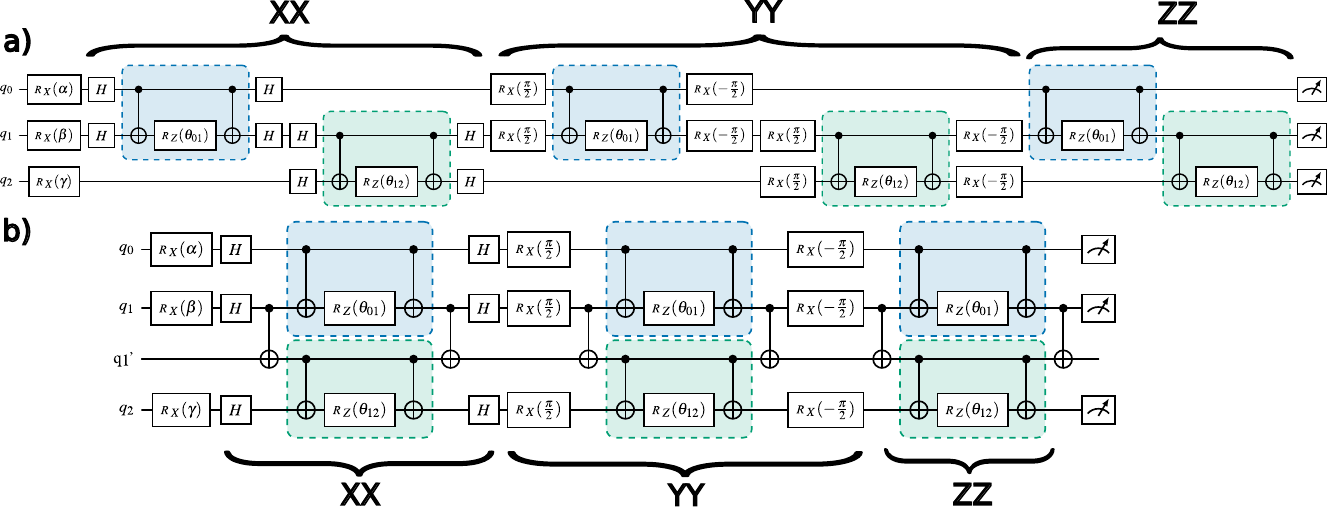}
    \caption{Not-parallelized (a) and parallelized (b) circuits for chain connectivity, see Eq.~\ref{pf_chain}. Colored blocks correspond to the interaction terms with coefficients $\theta_{kl}$ (colors in Fig.~\ref{fig:chain}).}
    \label{fig:chain_circs}
\end{figure*}

\begin{equation}\label{pf_exp}
    \begin{split}
        &\exp\biggl[ -it\biggl( \sum_{k=1}^N \omega_{k}I^X_k + 2\pi\sum_{k<l}J_{kl}\vec{I}_k \cdot \vec{I}_l \biggr) \biggr] \approx \\
        &\approx \exp\biggl( -it\sum_{k=1}^N \omega_{k}I^X_k \biggr)\cdot\exp\biggl( -2it\pi\sum_{k<l}J_{kl}I^X_kI^X_l \biggr) \cdot \\
        &\cdot \exp\biggl( -2it\pi\sum_{k<l}J_{kl}I^Y_kI^Y_l \biggr) \cdot \exp\biggl( -2it\pi\sum_{k<l}J_{kl}I^Z_kI^Z_l \biggr)
    \end{split}
\end{equation}

When the Zeeman splittings $\omega_{k}$ are much larger than the couplings $J_{kl}$, as is the case in high-field liquid-state NMR, it is advantageous to work in the interaction picture, where the single-spin Zeeman rotations (the $X$ block in Eq.~\ref{pf_exp}) are absorbed into the classical frame and only the bilinear coupling terms ($XX$, $YY$, $ZZ$) remain in the Trotterized circuit. This reduces the circuit depth per step \cite{low_wiebe_2018, burov2026}. The numerical demonstration in Sec.~\hyperref[sec:error_detection]{Error detection and post-selection} takes place at zero field, where no Zeeman terms are present, and uses a second-order product formula with $M=32$ Trotter steps.

As a concrete example, consider a $3$-spin molecule with chain connectivity. As depicted in Fig.~\ref{fig:chain}, we map the 3 spins to qubits $\{ q_0, q_1, q_2 \}$ with respective Zeeman frequencies $\{ \alpha, \beta, \gamma \}$ and spin-spin coupling parameters $\theta_{01}$ and $\theta_{12}$. The corresponding first-order, one-step product formula is then given in Eq.~\ref{pf_chain}, which corresponds exactly to the circuit displayed in Fig.~\ref{fig:chain_circs}.

\begin{equation}\label{pf_chain}
    \begin{split}
        &\exp\biggl( -it(\alpha I^X_0 + \beta I^X_1 + \gamma I^X_2) \biggr)\cdot \\
        &\cdot\exp\biggl( -2it\pi(\theta_{01} I^X_0I^X_1 + \theta_{12} I^X_1I^X_2) \biggr) \cdot \\
        &\cdot\exp\biggl( -2it\pi(\theta_{01} I^Y_0I^Y_1 + \theta_{12} I^Y_1I^Y_2) \biggr) \cdot \\
        &\cdot\exp\biggl( -2it\pi(\theta_{01} I^Z_0I^Z_1 + \theta_{12} I^Z_1I^Z_2) \biggr)\\
    \end{split}
\end{equation}

Generally in this circuit, the different contributions within the bilinear coupling terms ($XX$, $YY$, $ZZ$) cannot be executed in parallel and the circuit depth scales as $O(|E|)$, where $E$ is the edge set of the interaction graph (i.e.\ the set of spin pairs with non-zero $J_{kl}$), which gives $O(N^2)$ for fully connected graphs.

\section{Parallelization/gadget}\label{sec:gadget}

Depth compression by trading additional qubits, measurements, or nonlocal fan-out resources has been explored in several settings \cite{doi:10.1137/S0097539799355053, v001a005, q418-pydy, BROADBENT20092489}, algebraic circuit identities give compression without added qubits \cite{PhysRevA.105.032420}, and logarithmic-depth improvements have been demonstrated with other quantum primitives such as the preparation of entangled states \cite{logGHZ}. Ancilla-based parallelization has been developed for Hamiltonian-simulation algorithms at the asymptotic level \cite{zhang_2024}, for commuting terms in the SELECT subroutine of linear-combination methods \cite{boyd_2023}, and as serial-to-parallel tradeoffs in fault-tolerant compilation \cite{litinski_2019, ppr_2026}; the first two act within linear-combination-of-unitaries and quantum-walk algorithms, Zhang et al.\ compressing the dependence of the depth on the target precision and Boyd diagonalizing commuting terms in a block-encoding multiplexer, whereas the present construction parallelizes a product-formula circuit along the spatial dimension of the interaction graph and supplies an error-detecting register and end-to-end routed-hardware evaluation absent from that line of work. Tunable width-depth tradeoffs are also known for diagonal-unitary synthesis on planar grids \cite{xu_2026}, and a recent compiler for spin and fermionic Hamiltonian simulation reduces depth by gate synthesis and partial Trotterization within a fixed qubit count \cite{kernpiler_2025}, complementary to the width-for-depth trade made here. The circuit fragment at the heart of our construction, a CNOT fan-out followed by parallel two-qubit rotations and the inverse fan-out, is a batched form of the phase gadgets of circuit synthesis \cite{cowtan_2020}; we use the term gadget in this compilation sense, distinct from the perturbative Hamiltonian gadgets of complexity theory \cite{kempe_2006}. Here we consider a system of $N$ spins governed by a liquid-state NMR Hamiltonian as defined in Eq.~\ref{NMRH2}.
In a standard quantum simulation mapping, each logical spin $k$ is assigned to a single physical qubit $q_k$. 
The implementation of interaction terms $U_{kl}(\theta) = \exp(-i \frac{\theta}{2} P_k P_l)$, where $P \in \{X, Y, Z\}$ denotes a Pauli operator (with $\theta$ absorbing the factor of $1/4$ from the spin-$\tfrac{1}{2}$ convention $I^P = \sigma^P/2$), is constrained by hardware topology. 
Specifically, two terms sharing a common logical spin, such as $U_{kl}$ and $U_{km}$, cannot be executed at the same time on a standard architecture because they require access to the same physical resource $q_k$.

\begin{figure*}[tp]
    \centering
    \includegraphics[width=\textwidth]{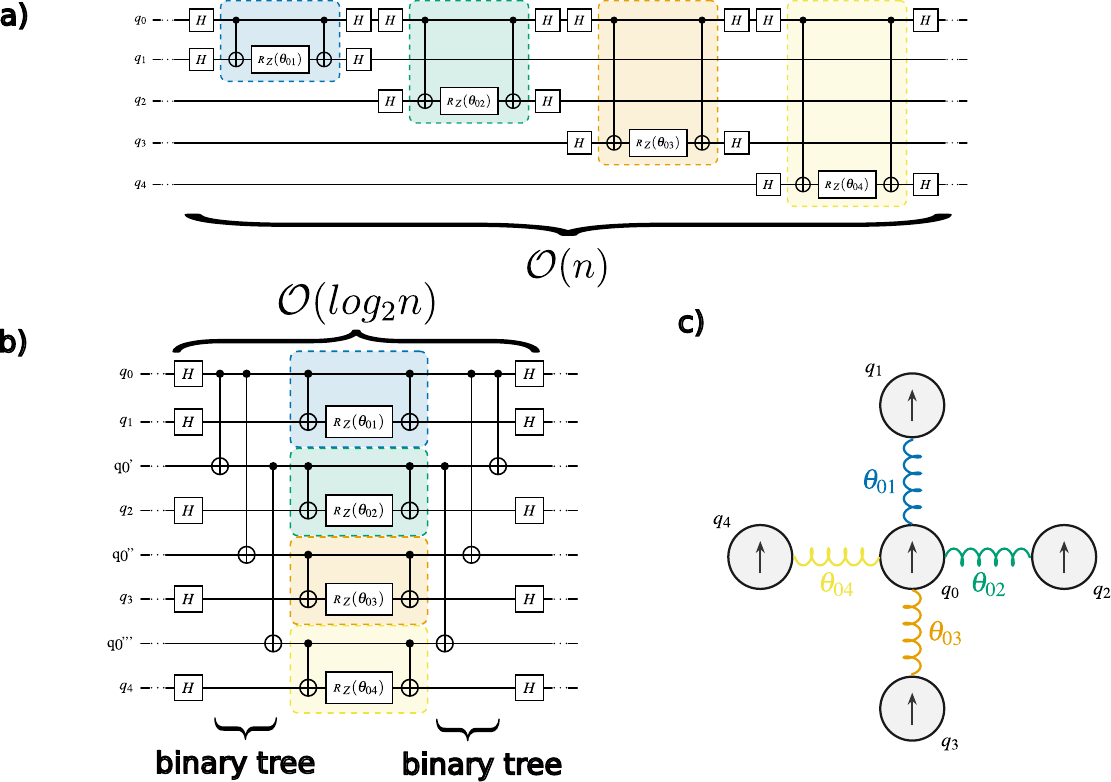}
    \caption{Four two-qubit operations that share a qubit, executed in sequence (a) and in parallel using the gadget (b), meant to be understood as a part of a larger quantum circuit. The corresponding star-topology molecule (c). Colored blocks correspond to the interaction terms with coefficients $\theta_{kl}$.} 
    \label{fig:bin_tree}
\end{figure*}

The gadget compiler (detailed pseudocode in Fig.~\ref{alg:gadget_parallel}) uses a redundant register mapping where each logical spin $k$ is represented by a cluster of $n_k$ physical qubits $\mathcal{Q}_k$, with $d_k$ the degree of spin $k$ in the interaction graph. In the maximally parallel form (full fan-out) $n_k = \max(1, d_k)$; the rounds parameter $R$ of Sec.~\ref{subsec:scheduled} generalizes this to $n_k = \max(1, \lceil d_k/R\rceil)$, down to a single qubit per spin:
\begin{equation}
    \mathcal{Q}_k = \{ q_{k,0}, q_{k,1}, \dots, q_{k,n_k-1} \}
\end{equation}
The set consists of a root qubit $q_{k,0}$ and $n_k-1$ ancillas.
At full fan-out the total number of physical qubits is $\sum_k n_k = \sum_k \max(1, d_k)$, which equals $2|E|$ when every spin participates in at least one coupling. For a fully connected interaction graph ($|E| = N(N-1)/2$) this yields $O(N^2)$ qubits (see Fig.~\ref{fig:bin_tree}), whereas sparse graphs such as linear chains ($|E| = N-1$) require only $O(N)$; the scheduled gadget uses fewer, $\sum_k \max(1, \lceil d_k/R\rceil)$.
The simulation is structured as a sequence of $M$ Trotter steps, each subdivided into interaction blocks $b \in \{1, \dots, B\}$ corresponding to different Pauli bases $\{XX, YY, ZZ\}$. 
An interaction block is defined by the unitary sequence:
\begin{equation}
    \mathcal{U}_b = W_b^\dagger \left( \prod_{k=1}^N F_k^\dagger \right) U_{b, \text{parallel}} \left( \prod_{k=1}^N F_k \right) W_b
\end{equation}
where $W_b$ is the local basis transformation (e.g., $W=H^{\otimes N}$ for $XX$ terms), $F_k$ is the fan-out operator, and $U_{b, \text{parallel}}$ is the parallel interaction layer.

The operator $F_k$ is implemented as a binary tree of CNOT gates, as visualized in Fig.~\ref{fig:bin_tree}, which prepares the entangled logical state:
\begin{equation}
    \ket{\Psi}_k = F_k \ket{\psi}_{q_{k,0}} \ket{0}^{\otimes (n_k-1)} = \alpha \ket{0}^{\otimes n_k} + \beta \ket{1}^{\otimes n_k}
\end{equation}
Each logical edge $\{k, l\} \in E$ is mapped to a unique pair of physical qubits $(q_{k,a}, q_{l,b})$ from the allocated clusters. 
Because each physical qubit is involved in at most one interaction gate per time-bin, all terms of the same Pauli type can be executed in parallel. 
Mathematical equivalence is maintained because a $Z$-basis interaction applied to any member of the repetition-coded gadget is identical to an operation on the logical root. 
The block concludes with $F_k^\dagger$, which contracts the cluster to the root, followed by $W^\dagger$. 
This reduces the depth of the coupling layer from $O(|E|)$ to $O(\log \Delta)$, where $\Delta$ is the maximum degree of the graph.

\begin{figure*}[tp]
\small
\hrule
\vspace{3pt}
\begin{algorithmic}[1]
\Procedure{Compile}{$S, G=(V, E), R$}
    \State \textbf{Input:} Pauli terms $S$, interaction graph $G$, rounds $R$
    \State \textbf{Output:} Parallelized Quantum Circuit $QC$
    \State \Comment{Initialization and Physical Mapping}
    \For{each logical qubit $v \in V$}
        \State $d_v \gets \text{degree}(v)$ in $G$
        \State Allocate $n_v = \max(1, \lceil d_v / R \rceil)$ physical qubits for $v$
        \State Designate one as $v_{root}$ and others as $v_{copies}$
    \EndFor
    \State Map each edge $e = \{u, v\} \in E$ to a unique pair $(u_i, v_j)$ from allocated sets
    \State Partition $E$ into rounds $E_1, \dots, E_\rho$ by earliest-fit greedy $f$-coloring (each spin $v$ in $\le n_v$ edges per round; exact for stars, $\rho \le R+1$ not guaranteed on dense graphs, App.~B)

    \State $B \gets$ group $S$ into blocks of terms with same Pauli type
    \For{each block $b \in B$}
        \State $\text{type} \gets$ Pauli type of block $b$
        \If{$\text{type} = X$} \Comment{Single-spin Zeeman terms}
            \For{each term $(X, \{v\}, \theta) \in b$}
                \State Apply rotation $R_X(\theta)$ to $v_{root}$
            \EndFor
        \Else
            \State \Comment{Step 1: Basis Rotation}
            \For{each $v \in V$}
                \If{$\text{type} = XX$} apply $H$ to $v_{root}$
                \ElsIf{$\text{type} = YY$} apply $R_X(\pi/2)$ to $v_{root}$
                \EndIf \Comment{$ZZ$: no rotation needed}
            \EndFor

            \State \Comment{Step 2: Fan-Out (Binary Tree)}
            \For{each $v \in V$}
                \State Apply binary CNOT tree from $v_{root}$ through $v_{copies}$ \Comment{depth $O(\log n_v)$}
            \EndFor

            \State \Comment{Step 3: Parallel interaction in $\rho$ rounds}
            \For{$r = 1$ to $\rho$}
                \For{each edge $\{u, v\} \in E_r$}
                    \State $(u_i, v_j) \gets$ physical mapping for $\{u, v\}$
                    \State Apply $CX(u_i, v_j) \to R_z(\theta_{uv}, v_j) \to CX(u_i, v_j)$
                \EndFor
            \EndFor

            \State \Comment{Step 4: Fan-In (Inverse Binary Tree)}
            \For{each $v \in V$}
                \State Apply reverse-order binary CNOT tree to contract $v_{copies}$ to $v_{root}$
            \EndFor

            \State \Comment{Step 5: Inverse Basis Rotation}
            \For{each $v \in V$}
                \If{$\text{type} = XX$} apply $H$ to $v_{root}$
                \ElsIf{$\text{type} = YY$} apply $R_X(-\pi/2)$ to $v_{root}$
                \EndIf
            \EndFor
        \EndIf
    \EndFor
    \State \Return $QC$
\EndProcedure
\end{algorithmic}
\vspace{3pt}
\hrule
\caption{Scheduled gadget compiler. The rounds parameter $R$ sizes each spin's register to $n_v=\max(1,\lceil d_v/R\rceil)$ physical qubits, a root and $n_v-1$ copies, and interpolates between the full fan-out ($R=1$, register size equal to the spin's degree) and the sequential circuit ($R=\Delta$, one qubit per spin). Terms are grouped into same-type Pauli blocks; each $XX$ or $YY$ block is conjugated by a single-spin basis rotation on the roots (Steps~1 and 5), so within the block every interaction reduces to a $ZZ$ rotation on the copies, while a $ZZ$ block needs no rotation. Within a block the registers are fanned out once by binary CNOT trees (Step~2, depth $O(\log n_v)$), the interactions run in $\rho$ parallel rounds set by an $f$-coloring of the interaction graph ($\rho=\chi'_f(G)\le R+1$, attained by the earliest-fit greedy for the star systems here; Step~3), and the registers are fanned in once (Step~4), so the fan-out and fan-in trees are paid once per block instead of once per round, whatever the value of $R$.}
\label{alg:gadget_parallel}
\end{figure*}

\subsection{Scheduled parallelization}\label{subsec:scheduled}

The full fan-out (Fig.~\ref{alg:gadget_parallel} with $R=1$) allocates $\max(1, d_k)$ qubits to spin $k$ and executes all of its interactions in a single parallel layer. Between this extreme and the sequential circuit lies a family of schedules indexed by the rounds parameter $R$, which fixes the register sizes: spin $k$ is fanned out into a register of $n_k = \lceil d_k / R \rceil$ qubits, from a single qubit at $R=\Delta$ (the sequential circuit) to one register qubit per interaction at $R=1$ (the full fan-out), trading width against depth. The interactions then run in parallel rounds, each spin using at most its $n_k$ register qubits at a time, so a round is an $f$-coloring class of the interaction graph with capacities $n_k$, and the number of such rounds $\rho$ is bounded in Remark~\ref{rem:rho}. The terms are grouped by Pauli type, and within each type block the register is fanned out once, all $\rho$ interaction rounds run on the resulting copies, and the register is fanned in once. The copies persist across every round, so the fan-out and fan-in CNOT trees are paid once per block instead of once per round, whatever the value of $R$; interleaving the rounds across Pauli types or contracting the register between rounds would repeat them. The optimal $R$ balances this one-off fan-out cost against the parallelism gained per round.

\begin{remark}[round count]\label{rem:rho}
The number of interaction rounds is the degree-constrained ($f$-)chromatic index $\rho=\chi'_f(G)$ of the interaction graph $G$ with capacities $n_k$: the fewest classes into which the interactions partition with each spin in at most $n_k$ per class. Fanning out realizes the balanced vertex split that carries this $f$-coloring to an ordinary edge coloring of a simple graph of maximum degree $\Delta_f=\max_k\lceil d_k/n_k\rceil\le R$, and the $f$-coloring form of Vizing's theorem~\cite{hakimi_kariv_1986} (extended to multigraphs in Ref.~\cite{nakano_nishizeki_1988}) bounds it,
\begin{equation}
  \Delta_f \;\le\; \rho \;\le\; \Delta_f + 1 \;\le\; R+1 .
  \label{eq:rho_bound}
\end{equation}
The busiest register sets the lower bound $\Delta_f$, at most $R$ and below it when a register is over-provisioned; the extra round is the Vizing gap of the split, absent when $G$ is bipartite, the $f$-analog of K\"onig's theorem~\cite{konig_1916}, and present for some non-bipartite graphs, where deciding it is NP-hard already for ordinary edge coloring~\cite{holyer_1981}. The endpoints: $R=\Delta$ gives the split $G$ itself, so $\rho=\chi'(G)\in\{\Delta,\Delta+1\}$ by ordinary Vizing~\cite{vizing_1964}; $R=1$ gives each interaction its own copy, so $\rho=1$. The star systems studied here, the demonstration of Table~\ref{tab:resources} and the synthetic ladder of Appendix~\hyperref[app:stars]{D}, are bipartite stars, so $\rho=\Delta_f\le R$ with equality when $R\mid D$, and the earliest-fit greedy of Appendix~\hyperref[app:stats]{B} attains this optimum for them; for the multi-hub and cluster systems of Table~\ref{tab:molecules} the same greedy sets both the gadget and baseline round counts and need not reach the $f$-chromatic index.
\end{remark}

We can make both the depth compression and the width--depth trade-off precise for the star geometry of the demonstration and of the largest gains, a hub of degree $D$ coupled to $D$ leaves.

\begin{proposition}[star graph]\label{prop:staropt}
For a hub of degree $D$ coupled to $D$ leaves and rounds parameter $R$, in the routing-free limit and up to single-spin rotations, the two-qubit depth of a Pauli block is at most $2\lceil\log_2\lceil D/R\rceil\rceil + 2R$ (with equality when $R$ divides $D$), which at full fan-out ($R=1$) equals $2\lceil\log_2 D\rceil + 2 = O(\log D)$, a logarithmic compression of the sequential depth $2D$. Its width-depth product, the two-qubit volume,
\begin{equation}
V(R) = \bigl(D + \lceil D/R\rceil\bigr)\bigl(2\lceil\log_2\lceil D/R\rceil\rceil + 2R\bigr),
\label{eq:starvol}
\end{equation}
is minimized at $R^\ast = O(\log D)$ with $V(R^\ast) = O(D\log D)$; against the sequential volume $V(D) = 2D(D+1) = \Theta(D^2)$ the gadget lowers the volume for all sufficiently large $D$, with $V(R^\ast)/V(D) = O(\log D / D)$.
\end{proposition}
\begin{proof}
The width is $D$ leaves and $\lceil D/R\rceil$ hub-register qubits. The two-qubit depth of a Pauli block is a fan-out and a fan-in binary CNOT tree, each of depth $\lceil\log_2\lceil D/R\rceil\rceil$, enclosing at most $R$ interaction rounds of two CNOT layers each; their product is Eq.~\eqref{eq:starvol}, with equality when $R$ divides $D$. At $R=D$ the copies and the trees vanish, giving $V(D) = (D+1)(2D)$. At $R=2$ the width is at most $\lceil 3D/2\rceil$ and the depth is $2\lceil\log_2\lceil D/2\rceil\rceil + 4 = O(\log D)$, so $V(R^\ast)\le V(2) = O(D\log D)$. Since the width is at least $D$ and the depth at least $2R$, one has $V(R)\ge 2DR$; combined with $V(R^\ast)\le V(2)$ this forces $R^\ast = O(\log D)$. Finally $V(R^\ast)/V(D) = O(D\log D)/\Theta(D^2) = O(\log D / D)$.
\end{proof}

Numerically the minimizer is very flat, staying at $R^\ast \in \{2,3,4\}$ from $D=12$ to $D=100$, so the volume-optimal schedule uses a small, essentially degree-independent number of rounds and $\Theta(D)$ copies. The count is basis-independent: the ``up to single-spin rotations'' qualifier makes it invariant across the CNOT, CZ, and M{\o}lmer--S{\o}rensen bases, which are locally equivalent, so the depth and volume scalings carry over to hardware and only $O(1)$ constants change: an interaction rotation costs two entangling gates in the CNOT or CZ basis but only one native $R_{XX}$, as the two architectures of Table~\ref{tab:resources} bear out.

The star is the extremal case of an arbitrary interaction graph, with the degree inhomogeneity setting the gain.

\begin{proposition}[general graph]\label{prop:general}
For an interaction graph on $N$ spins with maximum degree $\Delta$ and average degree $\bar d = 2|E|/N$, the scheduled gadget in the routing-free limit has width $\sum_v \max(1,\lceil d_v/R\rceil)$ and Pauli-block two-qubit depth $2\lceil\log_2\lceil\Delta/R\rceil\rceil + 2\rho$, with $\rho=\chi'_f(G)\le R+1$ the $f$-chromatic index of Remark~\ref{rem:rho}, the fewest rounds the schedule admits. At full fan-out the depth is $O(\log\Delta)$, a logarithmic compression of the edge-colored sequential depth $\Theta(\Delta)$, and the width-depth volume obeys $V/V_{\mathrm{seq}} = O\bigl((\bar d/\Delta)\log\Delta\bigr)$. The gadget therefore lowers the volume for degree-inhomogeneous graphs, $\bar d = o(\Delta/\log\Delta)$, the star ($\bar d=\Theta(1)$) saturating the envelope $O(\log\Delta/\Delta)$ of Proposition~\ref{prop:staropt}.
\end{proposition}
\begin{proof}
Spin $v$ is encoded in $\max(1,\lceil d_v/R\rceil)$ register qubits, so the width is $\sum_v\max(1,\lceil d_v/R\rceil)$; for the star this is $\lceil D/R\rceil$ hub-register qubits together with the $D$ single-qubit leaves, i.e.\ $D+\lceil D/R\rceil$, recovering Proposition~\ref{prop:staropt}.

Within a Pauli block each register is fanned out and back once by binary CNOT trees; the trees of different spins act on disjoint qubits in parallel, so their combined depth is set by the largest register, $2\lceil\log_2\lceil\Delta/R\rceil\rceil$ CNOT layers. Between them the interactions partition into $\rho=\chi'_f(G)\le R+1$ rounds (Remark~\ref{rem:rho}); the earliest-fit greedy attains this for stars and may exceed it on dense graphs (Appendix~\hyperref[app:stats]{B}). Each round is two CNOT layers, so the block depth is $2\lceil\log_2\lceil\Delta/R\rceil\rceil+2\rho$; at full fan-out ($R=1$) the trees have depth $2\lceil\log_2\Delta\rceil$ and one round remains, giving $2\lceil\log_2\Delta\rceil+2=O(\log\Delta)$.

The $R=\Delta$ endpoint keeps one qubit per spin, with round count the chromatic index $\chi'(G)=\Theta(\Delta)$ (Remark~\ref{rem:rho}); the earliest-fit greedy realizes this optimum for stars and stays $\Theta(\chi'(G))$ in general (Appendix~\hyperref[app:stats]{B}). With two CNOT layers per round and width $N$, its volume is $V_{\mathrm{seq}}=2N\chi'(G)=\Theta(N\Delta)$.

At $R=1$ the gadget width is $\sum_v d_v=2|E|=N\bar d$ and its depth is $2\lceil\log_2\Delta\rceil+2$, so $V(1)=2N\bar d\,(\lceil\log_2\Delta\rceil+1)$ and
\[
\frac{V(1)}{V_{\mathrm{seq}}}=\frac{\bar d\,(\lceil\log_2\Delta\rceil+1)}{\chi'(G)}=O\!\Bigl(\tfrac{\bar d}{\Delta}\log\Delta\Bigr).
\]
Since the family contains both the full fan-out and the sequential circuit, the volume-optimal schedule obeys $V(R^\ast)\le\min\{V(1),V_{\mathrm{seq}}\}$. The ratio above is $o(1)$ precisely when $\bar d=o(\Delta/\log\Delta)$, and the full fan-out then already improves on the sequential circuit.
\end{proof}

Conversely the bound exceeds one for degree-regular graphs, where no member of the family improves on the sequential circuit; the molecules of Table~\ref{tab:molecules} and the star ladder of Appendix~\hyperref[app:stars]{D} interpolate between these limits, their distance from the star envelope measuring how much off-hub coupling density dilutes the gain.

The compression assumes the fan-out tree can be embedded, which needs branch points of coordination at least three; a degree-two line or ring has none and degrades the tree to linear depth, while heavy-hex, mostly degree two but branching at its degree-three junctions, hosts it. On real hardware the routed cost adds to Eq.~\eqref{eq:starvol}, and the resulting gain is largest at an intermediate coordination that matches the fan-out to the device topology (Appendix~\hyperref[app:connectivity]{E}).

\subsection{Resource comparison protocol}

This section reports the resource comparison between the sequential and scheduled-gadget compilations under a fixed protocol. The protocol fixes everything except the compilation: the identical term sequence (second-order product formula, merged step boundaries, barriers removed before transpilation) is compiled by each strategy and transpiled to two hardware targets: a heavy-hex superconducting target \cite{chamberland_2020}, the ibm\_aachen coupling map with basis $\{\mathrm{CZ}, \mathrm{RZ}, \mathrm{SX}, X\}$ and removal of unused wires, and an all-to-all trapped-ion target with the native M{\o}lmer--S{\o}rensen basis $\{R_{XX}, R_X, R_Y, R_Z\}$ \cite{sorensen_molmer_1999}. Transpilation uses Qiskit 2.3.0 \cite{qiskit_2024} at optimization level 3. Routing is the only stochastic element of the chain: the all-to-all results are seed-independent and are obtained from a single transpilation, while every heavy-hex cell is transpiled with 1024 seeds and reported as the median with a bootstrap 95\% confidence interval. For each cell we report the qubit count $w$, the two-qubit depth $d_{2q}$, the two-qubit gate count $n_{2q}$, and the two-qubit volume $V_{2q} = w \cdot d_{2q}$; in a small fraction of seeds the router recruits qubits beyond the allocated $w$, which does not affect the medians. The sequential baseline is the $R=\Delta$ endpoint of the same family, in which every spin keeps a single qubit and the interactions are edge-colored into layers, so that the comparison is between two limits of one construction rather than two unrelated compilations. Both limits share one scheduling primitive: the interactions of a Pauli block are partitioned into the fewest rounds allowed by the register sizes, each spin $k$ appearing in at most $n_k = \lceil d_k/R \rceil$ of them per round. Scheduling commuting two-qubit interactions into parallel layers by edge coloring the interaction graph is a standard compilation primitive \cite{guerreschi_park_2018, lao_browne_2022}, applied to product-formula circuits to minimize the Trotter-layer depth \cite{bringewatt_davoudi_2023}; here it is the degree-constrained ($f$-)coloring \cite{hakimi_kariv_1986, nakano_nishizeki_1988} whose per-vertex capacity $n_k$ is set by the fan-out register size, so that a single parameter interpolates between the proper edge coloring at $R=\Delta$ (the sequential baseline \cite{misra_gries_1992}, against which any gain reflects parallelism alone) and the full fan-out at $R=1$. The round count is the $f$-chromatic index of Remark~\ref{rem:rho}; we compute the coloring by earliest-fit greedy (Appendix~\hyperref[app:stats]{B}), which for the star systems, where the gain is largest, reduces to the trivial round-robin of the hub and is optimal; on denser systems the same greedy sets both the gadget and the baseline round counts and need not reach the optimum. The resource ratios depend on the Trotter step count only through the step boundaries: in the routing-free limit the circuit at $M$ steps is $M$ repetitions of the step circuit up to boundary mergers, so the all-to-all counts are exactly affine in $M$ and the ratios approach their per-step values as $1/M$. The $R=3$ depth ratio is $0.517$, $0.503$, and $0.501$ at $M=1$, $8$, and $32$, the last within $10^{-3}$ of the asymptotic $1/2$; the heavy-hex $R=3$ depth ratios of this star at the same step counts are $0.476$, $0.485$, and $0.481$ with overlapping confidence intervals. The heavy-hex ratios reported here and in Table~\ref{tab:molecules} are medians over 1024 seeds with bootstrap 95\% confidence intervals, and the all-to-all ratios are exact (Appendix~\hyperref[app:stats]{B}). Table~\ref{tab:resources} presents the comparison for the 13-spin star example of Sec.~\hyperref[sec:error_detection]{Error detection and post-selection} (hub degree $D=12$) at $M=32$.

\begin{table*}[tp]
\begin{tabular}{lccccccc}
\toprule
 & & \multicolumn{3}{c}{Heavy-hex (ibm\_aachen)} & \multicolumn{3}{c}{All-to-all (native $R_{XX}$)} \\
Compilation & $w$ & $d_{2q}$\,[95\% CI] & $n_{2q}$ & $V_{2q}$ ratio & $d_{2q}$ & $n_{2q}$ & $V_{2q}$ ratio \\
\midrule
sequential & 13 & 6448 [6378, 6455] & 10812 & 1.00 & 1548 & 1548 & 1.00 \\
$R=1$      & 24 & 4580 [4552, 4632] & 13992 & 1.31 & 1033 & 4258 & 1.23 \\
$R=2$      & 18 & 3354 [3313, 3362] & 9713  & 0.72 & 904  & 2710 & 0.81 \\
$R=3$      & 16 & 3103 [3097, 3113] & 9108  & 0.59 & 775  & 2194 & 0.62 \\
$R=4$      & 15 & 3430 [3358, 3489] & 8893  & 0.61 & 904  & 1936 & 0.67 \\
$R=6$      & 14 & 3996 [3972, 4000] & 8671  & 0.67 & 904  & 1678 & 0.63 \\
$R=12$     & 13 & 6448 [6378, 6455] & 10812 & 1.00 & 1548 & 1548 & 1.00 \\
\bottomrule
\end{tabular}
\caption{Transpiled resource counts for the 13-spin star example (hub degree $D=12$, second-order product formula with $M=32$ Trotter steps). Columns give the qubit count $w$, the two-qubit gate depth $d_{2q}$, the two-qubit gate count $n_{2q}$, and the two-qubit volume $V_{2q} = w \cdot d_{2q}$. Heavy-hex entries are medians over 1024 transpiler seeds; the [95\% CI] after each heavy-hex depth is the bootstrap 95\% confidence interval of the median, and the confidence half-widths of the volume ratios are at most 0.021. The all-to-all entries are exact because transpilation is deterministic without routing. The volume ratio uses the median depth and is relative to the sequential compilation of the same architecture, the $R=\Delta$ endpoint of the same family. $R=12$ reproduces the sequential circuit exactly; $R=8$ allocates the same registers as $R=6$ and produces identical circuits (omitted).}
\label{tab:resources}
\end{table*}

The full fan-out is not the optimum: at $R=1$ the two-qubit volume grows by a quarter to a third across the two architectures, and the scheduling freedom turns the construction into a net win at this hub degree. On the heavy-hex target the optimum sits at $R=3$ with bootstrap probability above 0.99, the two-qubit depth rising on either side of it as the fan-out cost trades against the round count, and the gadget halves the two-qubit depth (ratio 0.481, 95\% confidence interval $[0.480, 0.487]$) while also reducing the two-qubit gate count (0.84): the parallel interaction layers map to disjoint physical edges and save routing overhead in addition to depth. On all-to-all connectivity the depth also halves (0.50 at $R=3$) while the gate count rises (1.42), since there is no routing to save and the fan-out CNOTs are purely additional gates; the volume optimum of 0.62 sits at $R=3$, with $R=6$ close behind. Because the volume weighs depth against width rather than the total gate count, this all-to-all gain records a benefit on depth-limited hardware; where the two-qubit gate count is the binding cost, as on gate-error-dominated trapped-ion devices, the increased count can offset the depth reduction. The $R=12$ row reproducing the sequential circuit is an internal consistency check of the protocol.

\begin{table*}[tp]
\begin{tabular}{lcccccccccc}
\toprule
 & \multicolumn{5}{c}{Heavy-hex (ibm\_aachen)} & \multicolumn{5}{c}{All-to-all (native $R_{XX}$)} \\
System ($N$, $\Delta$) & $R$ & $w$ & $d_{2q}$ & $n_{2q}$ & $V_{2q}$ & $R$ & $w$ & $d_{2q}$ & $n_{2q}$ & $V_{2q}$ \\
\midrule
Tetramethylsilane (13, 12)   & 3  & 16 & 0.48 & 0.84 & 0.59 & 3 & 16 & 0.50 & 1.42 & 0.62 \\
HMPA star model (19, 18)     & 6  & 21 & 0.51 & 0.78 & 0.56 & 6 & 21 & 0.50 & 1.17 & 0.55 \\
Tetraethylsilane (21, 20)    & 12 & 22 & 0.59 & 1.03 & 0.61 & 6 & 24 & 0.40 & 1.11 & 0.46 \\
Two-hub phosphine model (22, 12)& 6  & 24 & 0.63 & 0.86 & 0.69 & 6 & 24 & 0.58 & 1.09 & 0.64 \\
Phosphorus cluster (34, 15)  & 15 & 34 & 1.00 & 1.00 & 1.00 & 6 & 40 & 0.53 & 1.19 & 0.63 \\
Difluoroheptane (16, 6)      & 6  & 16 & 1.00 & 1.00 & 1.00 & 3 & 26 & 0.57 & 1.30 & 0.93 \\
\bottomrule
\end{tabular}
\caption{Volume-optimal schedules at $M=32$ Trotter steps for the demonstration and further spin systems: the 13-spin tetramethylsilane (TMS) star of Table~\ref{tab:resources}, the 19-spin star model of hexamethylphosphoramide (HMPA), tetraethylsilane, a constructed symmetric two-hub phosphine model, and further systems from the classical simulation literature \cite{burov2025largecircuitexecutionnmr, SPINACH_HOGBEN}. $N$ is the number of spins and $\Delta$ the maximum interaction-graph degree; $R$ is the rounds parameter of the quoted schedule, $w$ the qubit count, and $d_{2q}$, $n_{2q}$, $V_{2q}$ the two-qubit gate depth, gate count, and volume ($V_{2q} = w \cdot d_{2q}$). For each architecture the member of the schedule family with the smallest median $V_{2q}$ is quoted; the family contains the sequential circuit, so the optimum never loses; a row at $R=\Delta$ with unit ratios marks a system whose optimum is the sequential circuit itself. Ratios are relative to the sequential baseline of the same architecture, the $R=\Delta$ endpoint of the same family (heavy-hex: medians over 1024 transpiler seeds, confidence half-widths at most 0.02, except 0.032 for the tetraethylsilane volume ratio; all-to-all: exact). The gate-count ratio $n_{2q}$ may exceed one, since the volume, not the gate count, is the selection objective.}
\label{tab:molecules}
\end{table*}

Table~\ref{tab:molecules} extends the comparison to further spin systems at $M=32$. Every system gains on the all-to-all target, where the only requirement is a high-degree hub to parallelize: the volume reduction deepens with the maximum degree $\Delta$, from difluoroheptane ($\Delta=6$, 0.93) to tetraethylsilane ($\Delta=20$, 0.46). Heavy-hex imposes a second requirement, that the fan-out register embed as a localized patch of the lattice so its parallel interactions occupy disjoint physical edges. A system whose interaction graph embeds as one localized register meets it: tetramethylsilane, the HMPA star model, and tetraethylsilane, each a single high-degree hub, and the compact two-hub phosphine model, whose two adjacent hubs and their arms still fit a single localized patch, all gain on heavy-hex (0.56 for HMPA, 0.61 for tetraethylsilane). Systems with several mutually coupled high-degree centers do not. The phosphorus cluster separates the two requirements: its degree inhomogeneity earns an all-to-all gain of 0.63, yet its densely coupled seven-phosphorus core (three degree-15 hubs and four lower-degree centers, all mutually coupled) cannot be embedded together in the sparse heavy-hex lattice, so routing serializes the parallel layer and the volume optimum falls back to the sequential circuit. Difluoroheptane, low-degree and near-uniform, falls back on heavy-hex for the opposite reason, having no hub to fan out; forcing a fan-out schedule on such a system is counterproductive, its best fan-out schedule increasing the two-qubit volume rather than reducing it. Synthetic star systems place the onset of the heavy-hex gain at hub degree five and improve it to 0.30 at degree 60 (Appendix~\hyperref[app:stars]{D}). The resource ratios depend on the step count only through the $1/M$ boundary effect quantified above.

\section{Error detection and post-selection}\label{sec:error_detection}

This architecture detects errors without any algorithmic overhead. Redundancy introduced for a computational purpose that is simultaneously an error-detecting code is known from parity-encoded architectures for combinatorial optimization \cite{lechner_2015, pastawski_2016}; here the same effect arises inside a product-formula simulation circuit, with the checks read from the terminal measurement. After the parallelized gate layer the ancilla qubits are released (Fig.~\ref{fig:bin_tree}): following the basis rotation $W_b$, all interactions in the parallel layer decompose as $ZZ$ rotations (regardless of whether the original Pauli type was $XX$, $YY$, or $ZZ$), and these preserve the code subspace of the repetition code: the logical states $\ket{0}^{\otimes n_k}$ and $\ket{1}^{\otimes n_k}$ are eigenstates of $Z_i Z_j$, so the GHZ-like entangled state produced by the fan-out $F_k$ remains within the code space throughout the interaction layer and can be cleanly uncomputed by the inverse CNOT fan-in $F_k^\dagger$. 

The error detection uses a distance $d=2$ repetition code \cite{gottesman_1997} defined for each spin register $\mathcal{Q}_k$ of size $n_k$, meaning that any single-qubit bit-flip error can be detected but not corrected. 
The code space $\mathcal{C}_k$ is the $+1$ eigenspace of the stabilizer group $\mathcal{S}_k$, generated by:
\begin{equation}
    S_{k,j} = Z_{k,0} Z_{k,j} \quad \text{for } j \in \{1, \dots, n_k-1\}
\end{equation}
where $q_{k,0}$ is the root qubit and $\{q_{k,j}\}_{j=1}^{n_k-1}$ are the ancillas. 
The logical basis states are defined as $\ket{0}_L = \ket{0}^{\otimes n_k}$ and $\ket{1}_L = \ket{1}^{\otimes n_k}$. 

The full simulation circuit is composed of $M \times B$ interaction blocks (see Fig.~\ref{alg:gadget_parallel}). 
Each block implements a cycle of expansion $F_k$, parallel interaction $U_{\text{parallel}}$, and compression $F_k^\dagger$. 
The detectability of stochastic Pauli noise depends on how these errors map to the syndrome register following the application of $F_k^\dagger$.

\textbf{Bit-flip noise ($X$):}
Let $m \in \{0, 1\}^{n_k}$ be a vector denoting the bit-flips accumulated on the physical qubits of register $\mathcal{Q}_k$. 
The unitary decoder $F_k^\dagger$ is an invertible $\mathbb{F}_2$-linear map, so the measured syndrome $s \in \{0, 1\}^{n_k-1}$, read from the $n_k-1$ non-root copies after fan-in, is a fixed linear image of $m$. For a linear fan-in that XORs every copy directly against the root this image is $s_j = m_j \oplus m_0$; the binary CNOT tree actually compiled gives a different but equally invertible image. In either case the fan-in computes parities of the copies, so the syndrome vanishes exactly when every copy carries the same flip, and an error is undetected if and only if $m_0 = m_1 = \dots = m_{n_k-1}$. 
This condition is met only by the identity operation ($m = \vec{0}$) or the logical bit-flip $\bar{X} = X^{\otimes n_k}$ ($m = \vec{1}$). 

During the parallel interaction layer, any weight-1 bit-flip on a single physical qubit $q_{k,j}$ (whether root or ancilla) produces a non-trivial syndrome and is therefore detectable. In particular, a flip on the root ($j=0$) yields $s_i = 0 \oplus 1 = 1$ for all $i$, while a flip on ancilla $j > 0$ yields $s_j = 1$ with all other syndrome bits zero. The interaction gates couple copies in different registers, so a control fault during an interaction propagates an additional bit-flip into the partner register; each register is checked by its own syndrome, so the fault is caught whenever a fanned-out register ($n_k > 1$) is involved, which holds for every interaction of a fanned-out spin.

During the binary tree expansion $F_k$ (see Fig.~\ref{fig:bin_tree}), the situation differs because a bit-flip on a control qubit propagates through subsequent CNOT gates to its descendants. If a bit-flip occurs on the root qubit $q_{k,0}$ at the start of the expansion, it propagates to every qubit in the register, resulting in $m = \vec{1}$ and an undetected logical failure. In contrast, any bit-flip occurring on an internal node or target qubit that does not propagate to the entire cluster results in $s \neq \vec{0}$ and is detectable. Errors during the fan-in $F_k^\dagger$ are analogous: a bit-flip on a copy propagates only to qubits earlier in the reverse tree order and generally does not reach the root, producing a detectable syndrome.

\textbf{Other errors:}
Phase-flip errors ($Z$) commute with the stabilizers $S_{k,j}$ and are not detectable by the syndrome register. Their effect depends on when they occur. Between interaction blocks the copies rest in $\ket{0}$, on which $Z$ acts trivially, and only the root qubit is vulnerable. During an interaction window, however, the register occupies the code space, where a $Z$ on any of the $n_k$ members acts as the logical $\bar{Z}$: the logical dephasing cross-section during the windows is $n_k$ times that of a single qubit, in exchange for the shorter schedule. In the NMR context these events contribute to $T_2$-like broadening. 
Leakage and readout errors are outside the scope of the $d=2$ stabilizer protection, though high-fidelity measurement limits their impact.

\textbf{Error propagation along blocks and post-selection:}
In the simplest case of a device without mid-circuit measurement the ancilla register is not reset between blocks, and any bit-flip error acquired in an early block preserves the parity mismatch relative to the root throughout the remaining unitary evolution: if a copy ends in $\ket{1}$ after one fan-in, subsequent fan-out cycles propagate this anti-correlation through the register, maintaining a non-trivial syndrome across all later blocks.
This allows for a single measurement of all ancillas at the end of the Trotter sequence. The root qubits carry the measurement signal (the NMR observable), while the copy qubits are the syndrome register. The measurement counts are processed via a global post-selection filter defined by the projector $\Pi_0 = \bigotimes_k \ket{\vec{0}}\bra{\vec{0}}_{anc,k}$. The expectation values are then computed over the null-syndrome shots alone.   

To probe the regime where the depth compression matters, we simulate the zero-field NMR spectrum of tetramethylsilane (TMS, \ce{^{29}Si(CH3)4}; Fig.~\ref{fig:tms_system}), a star-shaped 13-spin system in which $^{29}$Si couples identically to twelve protons ($^2J_{\mathrm{SiH}} = 6.6$~Hz \cite{tms_j_1990}). At zero field the Hamiltonian contains only the scalar couplings, $H = 2\pi J \sum_{k=1}^{12} \vec{I}_{\mathrm{Si}} \cdot \vec{I}_{k}$. We take the idealized $\mathrm{A}_{12}\mathrm{X}$ limit of twelve magnetically equivalent protons, in which the small proton--proton couplings commute with the Hamiltonian and the observable and leave the comb positions unchanged; they are omitted. The spectrum of such a system is a comb of lines at the frequencies $(K+\tfrac{1}{2})J$, where $K$ is the total proton angular momentum quantum number \cite{ledbetter_2011, blanchard_2016}: six lines between 9.9 and 42.9~Hz (Appendix~\hyperref[app:demo]{C}). The hub degree of twelve maximizes the parallelization gain of the gadget on the heavy-hex architecture (Fig.~\ref{fig:milestones}), and because the circuits depend on the couplings only through the products $J\,dt$, the conclusions carry over to any star system with rescaled acquisition parameters. A four-spin zero-field spectrum has previously been computed on trapped-ion hardware \cite{seetharam_2023}.

\begin{figure}[htbp]
    \centering
    \includegraphics[width=0.55\columnwidth]{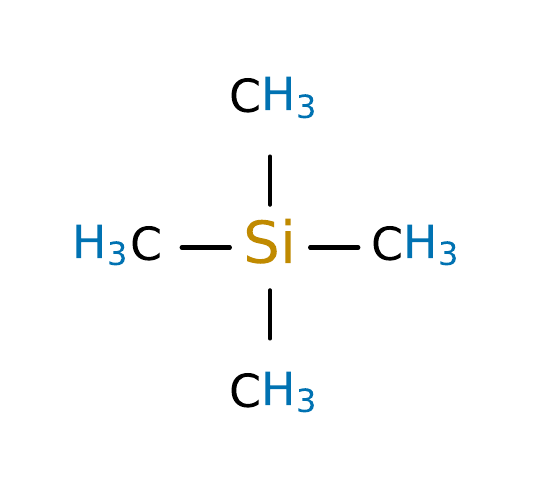}
    \caption{Tetramethylsilane (TMS). The $^{29}$Si nucleus couples identically to the twelve methyl protons ($^2J_{\mathrm{SiH}} = 6.6$~Hz), forming an interaction star of hub degree twelve. In the idealized $\mathrm{A}_{12}\mathrm{X}$ limit the proton--proton couplings commute with the Hamiltonian and the observable and are omitted.}
    \label{fig:tms_system}
\end{figure}

The simulation prepares the computational basis state with the silicon spin and five protons pointing up and seven protons pointing down. A single basis state requires no state-preparation circuit, and this pattern populates every sector $K \geq 1$ of the total proton spin, producing all six comb lines, while avoiding the $K=0$ sector, which holds a single $F=\tfrac{1}{2}$ level and contributes no line; the sector weights are given in Appendix~\hyperref[app:demo]{C}. The observable is the $\gamma$-weighted total magnetization $\sum_k (\gamma_k/\gamma_{\mathrm{H}})(\langle I^Z_k\rangle - i \langle I^Y_k\rangle)$, the quantity to which a zero-field magnetometer signal is proportional \cite{ledbetter_2011, blanchard_2016}. The Hamiltonian conserves every component of the total spin and the initial state is an eigenstate of the total $I^Z$, so the transverse expectation values vanish identically and the free induction decay is real; the measured imaginary quadrature is consistent with zero within shot noise and is discarded in the processing. Although not preparable in a magnetization-based experiment, this initial state is the natural state of the quantum simulation and shares with the physical zero-field experiment the complete comb of transition frequencies, which are set by the Hamiltonian alone. The dependence of the spectrum on the initial state is illustrated in Appendix~\hyperref[app:demo]{C}.

\begin{figure*}[tp]
    \centering
    \includegraphics[width=0.9\textwidth]{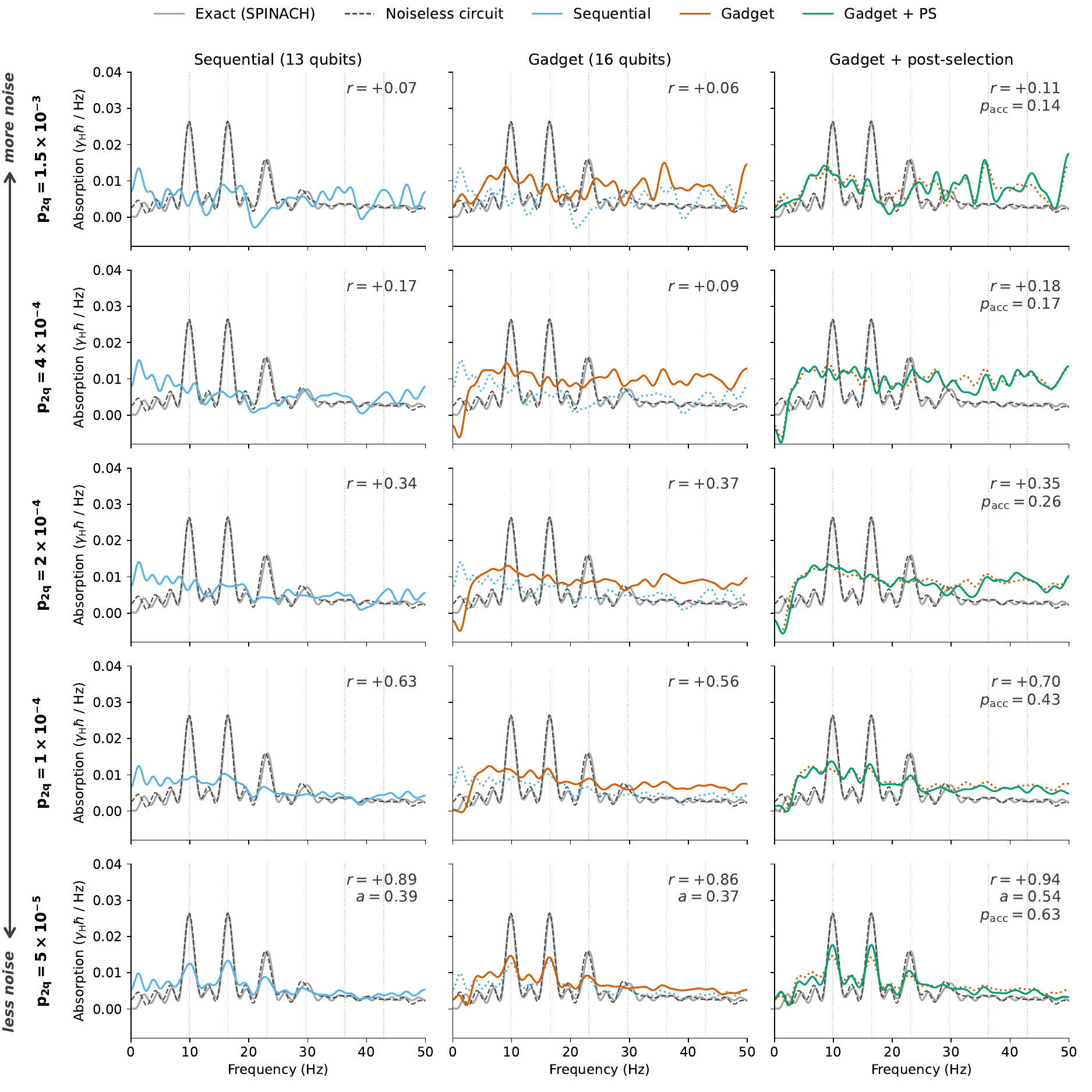}
    \caption{Simulated zero-field NMR spectra of tetramethylsilane (TMS,
$^{29}$Si(CH$_3$)$_4$): a 13-spin star system in which $^{29}$Si couples
identically to twelve protons ($^2J_{\mathrm{SiH}} = 6.6$~Hz, $B = 0$),
giving a comb of lines at frequencies $(K+\tfrac{1}{2}) \cdot 6.6$~Hz,
marked by faint vertical lines in every panel.
The time evolution is a second-order product formula with $M=32$ Trotter
steps; the free induction decay of the $\gamma$-weighted total
magnetization is sampled at 45 points over a 100~Hz spectral width with
4096 shots per point and Fourier transformed after mean subtraction and
1~Hz exponential line broadening. The three columns compare the two
compilations, each transpiled to a fixed subgraph of the ibm\_aachen
coupling map: the sequential compilation (13 qubits, two-qubit depth 4903), the gadget compilation at schedule $R=3$ (16 qubits, two-qubit depth 2452),
and the same gadget circuits with post-selection on the null ancilla
syndrome. Dotted curves repeat the trace of the preceding column. Rows
correspond to noise models in which the complete per-qubit ibm\_aachen
calibration (gate errors, coherence times, readout) is scaled uniformly
so that the median two-qubit error equals the quoted $p_{2q}$; the top
row is the unmodified calibration of the present-day device. Idle
relaxation during scheduled delays and readout errors are included.
Every panel carries two noise-free references: the exact spectrum,
computed independently with SPINACH, and the identical simulation
pipeline run at zero noise; the two nearly coincide; their residual difference is the $M=32$ Trotter error. All spectra share one
absolute scale in units of $\gamma_{\mathrm{H}}\hbar$ per Hz; no trace
is individually normalized. $r$ is the Pearson correlation with the
exact spectrum over the comb region 6--46~Hz; for the lowest-noise row,
the amplitude retention $a$, the slope of the linear fit
$S \approx a\,S_{\mathrm{exact}} + b$, gives the fraction of the exact
comb amplitude that survives the noise; $p_{\mathrm{acc}}$ is the
post-selection acceptance probability, the fraction of shots that pass
the null-syndrome filter. For $p_{2q} \le 1\times10^{-4}$ the post-selected gadget circuits reach, at half the two-qubit depth, spectral quality equal to or better than the sequential compilation; in the bottom row they reach $r = 0.94$ and $a = 0.54$, against $0.89$ and $0.39$.}
    \label{fig:noisy_spectra}
\end{figure*}

Both compilations implement a second-order product formula with $M=32$ Trotter steps, converged at the acquisition grid used here; the residual Trotter error is the small difference between the two reference curves in Fig.~\ref{fig:noisy_spectra}. Commuting single-spin layers at the boundaries of consecutive Trotter steps are merged and adjacent inverse gate pairs cancel; all quoted depths and gate counts refer to these simplified circuits. The sequential compilation places the 13 spins on 13 qubits. The gadget compilation allocates a four-qubit register to the hub (schedule $R=3$) and executes the twelve hub couplings in three parallel rounds of four interactions, for 16 qubits in total. Every acquisition point uses the same step count, so all circuits of one compilation share the same depth. The demonstration pins the circuits to a fixed calibrated 16-qubit subgraph of the ibm\_aachen coupling map, with the same basis and protocol as Table~\ref{tab:resources}; each compilation uses its volume-optimal transpiler draw, the seed of least two-qubit volume $V_{2q}$ over 1024 seeds, with ties broken by the lowest two-qubit gate count, then seed order (Appendix~\hyperref[app:stats]{B}). Volume is independent of the noise level, so the same circuit serves every scaling, and the error-budget distributions of the two compilations do not overlap across the seed pool, so no draw selection could affect the direction of the comparison. The selected sequential circuits have a two-qubit depth of 4903 with 10554 two-qubit gates, against 2452 and 8123 for the gadget: a depth compression by a factor of 2.0 obtained together with a 23\% reduction in two-qubit gate count for this draw, or 18\% at the median of the volume-optimal draws. These absolute counts fall below the full-map medians of Table~\ref{tab:resources} because the demonstration fixes a favorable 16-qubit subgraph and uses the volume-minimizing draw instead of the median seed; the depth ratio near 2.0 is consistent across both.

The noise model is constructed from the published per-qubit calibration of the ibm\_aachen device, a snapshot retrieved in July 2026 (Appendix~\hyperref[app:noise]{A}). Every gate carries thermal relaxation for its calibrated duration on the participating qubits, composed with a depolarizing residual chosen so that the total channel matches the calibrated gate error; readout errors are applied per qubit. At unit scale this parametric model reproduces the reference noise model that Qiskit constructs from the same calibration to within $3\times10^{-7}$ in average gate fidelity. The circuits are scheduled as soon as possible and every idle period carries thermal relaxation for its scheduled duration. The rows of Fig.~\ref{fig:noisy_spectra} divide all error rates by a common factor and multiply the coherence times by the same factor, so that each row is labeled by its median two-qubit error $p_{2q}$; the top row is the device as calibrated today. Each of the 45 acquisition points is sampled with 4096 shots per compilation by trajectory simulation \cite{qiskit_2024}.

The processing chain coincides with that of the SPINACH zero-field examples, mean subtraction, exponential apodization, and absorption-mode display, with the line broadening set to 1~Hz. Two independent checks anchor the simulation chain. The exact spectrum is recomputed with SPINACH from the same Hamiltonian, initial state, observable, and acquisition grid; the two exact free induction decays agree to $5.2\times10^{-6}$ per point, a residual set by the underlying gyromagnetic-ratio tables. Separately, the full noisy pipeline run at zero error reproduces the noiseless reference within shot noise at every acquisition point.

Figure~\ref{fig:noisy_spectra} shows the outcome. At the present-day calibration no compilation recovers any comb structure: the spectra are dominated by relaxation, and the post-selection acceptance of 0.14 sits at the random-parity floor of $2^{-3}$ expected from three fully scrambled ancilla bits. The comb begins to emerge near $p_{2q} = 2\times10^{-4}$ for all three traces, and from $p_{2q} = 1\times10^{-4}$ the post-selected gadget leads the shape recovery ($r=0.70$ against $0.63$ for the sequential compilation). In the lowest-noise row it reaches $r=0.94$ and retains $a=0.54$ of the exact comb amplitude, against $0.89$ and $0.39$ for the sequential compilation, while accepting 63\% of the shots. The structure below 5~Hz is the windowed-transform signature of the slowly varying damping; the sequential spectra carry the larger low-frequency pedestal, in line with their larger idle budget, and the metric band excludes this region.

Although the gadget circuit is twice as shallow and holds 23\% fewer two-qubit gates, its raw spectra trail the sequential ones in shape correlation at $p_{2q}=10^{-4}$ (0.56 against 0.63). The unfiltered average retains the flagged shots, more than half of which carry a detected fault at $p_{2q}=10^{-4}$ and contribute distorted spectra. A phase flip on any member of the encoded register during an interaction window also acts as the logical $\bar{Z}$, so the logical dephasing cross-section during the windows grows with the register size, and the shorter schedule repays this only in part. The depth advantage acts through the idle channel, visible as the smaller low-frequency relaxation pedestal of the gadget spectra, but at these error rates the gate channel dominates. Post-selection changes the balance: the same redundancy that raises the dephasing cross-section makes a large fraction of the faults detectable, and discarding the flagged shots makes the post-selected gadget the most accurate of the three traces at this noise level. At the lowest noise level the difference between the sequential and unfiltered-gadget correlations is within the shot noise of 4096 shots, whereas the improvement from post-selection is evaluated on the same shots. A bootstrap over the comb-band frequency points resolves it: at $p_{2q}=5\times10^{-5}$ the paired gains $r_{\mathrm{PS}}-r_{\mathrm{seq}}=+0.05$ [95\% CI $+0.04,\,+0.06$] and $a_{\mathrm{PS}}-a_{\mathrm{seq}}=+0.14$ [$+0.13,\,+0.15$] both exclude zero, as do $+0.07$ [$+0.05,\,+0.10$] and $+0.05$ [$+0.04,\,+0.07$] at the $10^{-4}$ crossover.

Post-selection carries a sampling cost. Keeping a fraction $p_{\mathrm{acc}}$ of the shots, matching the accepted-shot count of an unfiltered run takes $1/p_{\mathrm{acc}}$ times its wall-clock. The correlations above hold the budget fixed at 4096 shots per circuit and charge the discarded gadget shots against its own total, so this penalty is already reflected in them: the gadget accepts $p_{\mathrm{acc}}=0.63$ at $p_{2q}=5\times10^{-5}$ and $0.43$ at the $10^{-4}$ crossover, overheads of $1.6$ and $2.3$, and still leads. Since the gadget circuit is half as deep, each of its shots would also run in proportionally less gate time on hardware, which offsets part of this overhead where the per-shot runtime is dominated by gate depth. At the present-day $1.5\times10^{-3}$ the acceptance falls to near the $2^{-3}$ random-parity floor ($p_{\mathrm{acc}}=0.14$, a sevenfold overhead), but no compilation recovers the spectrum there, so the sampling cost is not the limiting factor.

\section{Conclusion/outlook}

We have presented a logarithmic-depth compression method for product-formula simulation of Heisenberg-type spin Hamiltonians based on fan-out parallelization. The method replaces a narrow, deep sequential implementation of pairwise spin interactions with a wider and shallower circuit in which each logical spin is represented by a degree-dependent register. A binary CNOT tree distributes the logical state across this register, allowing all interactions of the same Pauli type to be applied in parallel before the register is uncomputed. As a result, the depth of an interaction block scales logarithmically with the maximum degree of the interaction graph, rather than with the number of edges. The required qubit overhead is quadratic for fully connected systems but remains linear for sparse spin topologies such as chains.

The logarithmic depth gain demonstrated here is enabled by the binary CNOT fan-out tree, a primitive whose efficient preparation has been studied in the context of GHZ-state generation \cite{logGHZ} and whose computational power and circuit-depth advantages have been established in several theoretical and experimental settings \cite{doi:10.1137/S0097539799355053, v001a005, q418-pydy, BROADBENT20092489}. Our work shows that this primitive translates directly into a practical compiler-level optimization for Hamiltonian simulation.

Beyond depth reduction, the method enables post-selection against detectable errors without adding extra algorithmic layers. In noisy simulations of the 13-spin TMS system at zero field, under noise models anchored to the published calibration of a present-day quantum processing unit (QPU), the demonstration locates the error rates at which the spectrum of this example is recovered. At the current median two-qubit error of $1.5\times10^{-3}$ no compilation recovers the spectrum; recognizable combs appear near $2\times10^{-4}$; and from $10^{-4}$ the post-selected gadget circuits, at half the two-qubit depth, lead the sequential baseline in spectral quality ($r=0.70$ against $0.63$ at $10^{-4}$, and $r=0.94$ with $a=0.54$ against $0.89$ and $0.39$ at $5\times10^{-5}$). The register redundancy that enlarges the logical dephasing cross-section during the interaction windows is the same structure that flags a large fraction of the faults for discarding; post-selection turns the cost of the encoding into a net benefit.

The gadget compiler thus suits QPUs with spare qubits but a bounded depth budget, converting the register overhead into shallower interaction blocks and detectable-error post-selection. The fan-out parallelization principle extends to any product-formula simulation in which a qubit participates in multiple terms of the same Pauli type. This includes fermionic Hamiltonians under the Jordan-Wigner mapping \cite{jordan_wigner_1928}, where hopping terms decompose into blocks of uniform Pauli type that are individually parallelizable. The same degree criterion carries over (Proposition~\ref{prop:general}): the compression is largest where a mode couples to many others in the Pauli-type-resolved interaction graph, as in molecular electronic-structure Hamiltonians or higher-coordination lattice models, and it is absent for a degree-two one-dimensional chain such as the 120-qubit Fermi-Hubbard model recently simulated~\cite{hartnett2026fastaccuratehighresolutionsimulation}, which sits at the $O(\log 2)=O(1)$ no-compression end. For Hamiltonians with higher-locality or mixed-Pauli terms, the parallelization applies within each same-type block; the overhead analysis then depends on the Pauli-type-resolved degree of each qubit rather than the total interaction degree.

A natural next target is the central-spin physics of color-center registers. The NV$^-$ electron spin in diamond couples to tens of resolved $^{13}$C nuclear spins through individually measured hyperfine tensors, forming the high-degree star geometry in which the gadget compression is largest; the 27-nuclear-spin register of Ref.~\cite{abobeih_2019} is a named physical system of this kind. Every term of the free evolution commutes with the electron $Z$ operator, so the evolution between microwave pulses forms a single fan-out window and pulsed sequences arrive naturally partitioned for the gadget. Simulating polarization-transfer protocols such as PulsePol \cite{schwartz_2018} at these register sizes is demanding for classical methods and experimentally verifiable, which makes it a suitable first application. The construction itself also generalizes from the compiler presented here to a problem-agnostic compiler pass acting on any circuit in which many two-qubit interactions share a qubit; a systematic account is left to future work.

Classical NMR simulation packages typically operate in Liouville space to accommodate mixed states and relaxation, where unrestricted density-matrix evolution incurs $O(2^{2N})$ scaling with the number of coupled spins $N$ and becomes impractical around twenty spins \cite{SPINACH_HOGBEN, kuprovspin, das2025exploring, burov2025largecircuitexecutionnmr}; tensor network methods such as matrix product states (MPS) extend this reach to approximately 32 spins, the Liouville limit, where the broadly distributed entanglement of long-range-coupled NMR spins makes even accurate MPS evolution costly \cite{elenewski2024}. Quantum simulation of the coherent time evolution operates natively in Hilbert space ($2^N$), sidestepping the quadratic overhead in the exponent. In the near term, depth compression via the gadget compiler, combined with error suppression techniques \cite{burov2025largecircuitexecutionnmr}, may enable NISQ devices to push into this classically intractable regime. For instance, the 34-spin phosphorus cluster lies beyond the Liouville limit of density-matrix NMR packages, although the coherent zero-field evolution simulated here is a Hilbert-space ($2^N$) problem whose exploitable symmetry can lower the classical cost further, so the following is a comparison of compiled circuit sizes, not a demonstrated advantage. Its maximally parallel (full fan-out) schedule compiles to a two-qubit depth of 50 (one Trotter step) to 1290 (32 steps) at a width of 96 qubits, against 164 to 5124 at 34 qubits for the sequential compilation, a three- to fourfold depth reduction traded for width in the routing-free circuit. On all-to-all connectivity this depth reduction is also a volume reduction (0.63, Table~\ref{tab:molecules}); on the routed heavy-hex device the densely coupled phosphorus core cannot embed as a single localized register, so the volume-optimal schedule falls back to the sequential circuit, and the cluster illustrates that a width-for-depth volume gain need not appear on every architecture. As Fig.~\ref{fig:milestones} illustrates for the demonstration, whose transpiled circuit sizes are computed classically rather than executed on hardware, the trade shifts the requirement toward the wide-and-shallow regime where current superconducting QPUs offer sufficient qubit counts, and trapped-ion devices come close in both qubit number and demonstrated depth. The noisy demonstration of Sec.~\hyperref[sec:error_detection]{Error detection and post-selection} quantifies the remaining gap for the 13-spin system: median two-qubit error rates near $10^{-4}$, an order of magnitude below present calibrations, together with correspondingly longer coherence times, suffice to recover the spectrum of this 13-spin example.

\paragraph{\textbf{Data availability}}\label{Data}

The SPINACH script of Appendix~\hyperref[app:demo]{C}, the point sets of all figures and tables, the per-seed transpilation metrics of the demonstration-subgraph scan and of the Table~\ref{tab:molecules} endpoint scan, and the frozen calibration snapshot are available at \href{https://doi.org/10.5281/zenodo.22032137}{Zenodo (DOI: 10.5281/zenodo.22032137)}.

\paragraph{\textbf{Acknowledgment}}\label{Acknowledgment}
We acknowledge support from armasuisse Science and Technology (S+T), the Swiss Quantum Initiative (SQI) of the Swiss Academy of Sciences (SCNAT) and the State Secretariat for Education, Research and Innovation (SERI), as well as the National Centre of Competence in Research (NCCR) SPIN, funded by the Swiss National Science Foundation (grant number 51NF40-180604).

\paragraph{\textbf{Author contributions}}
 A.B. conceived the application of fan-out parallelization to NMR Hamiltonian simulation, implemented the compilers and simulations, ran the numerical experiments, and prepared the data and figures. All authors contributed to preliminary research, discussed the results, and prepared the manuscript. C.J. supervised the project.

\paragraph{\textbf{Use of AI tools}}
AI tools were used in an assistive capacity in preparing this manuscript: for language editing of author-written text; for restructuring and reformatting; for bibliographic research and citation verification; for assisting with and checking analysis and plotting code; and for numerical and symbolic verification and internal-consistency review. The authors conceived the research, wrote the manuscript, produced and independently verified all scientific results, and take full responsibility for the entire work. The AI tools used were Claude Code (Anthropic), with the Claude Opus 4.7 and Claude Opus 4.8 models.

\paragraph{\textbf{Competing interests}}
The authors declare no competing interests.

\clearpage

\bibliographystyle{quantum}
\bibliography{bibliography}

\clearpage

\appendix

\section{Noise model}\label{app:noise}

The noise model is generated from a frozen snapshot of the published ibm\_aachen calibration, retrieved in July 2026. Every calibrated gate carries thermal relaxation channels on its participating qubits, using the calibrated $T_1$, $T_2$, and gate duration, composed with a depolarizing residual sized so that the composite channel reproduces the calibrated gate error; where the thermal contribution alone exceeds the calibrated error, the residual is zero. Readout errors enter as symmetric classical bit flips at the calibrated rates. The circuits are scheduled as soon as possible with the calibrated durations, and every idle interval carries a thermal relaxation channel; idle durations are grouped into fifteen geometric bins of ratio 1.3, which rounds any duration by at most 13\%. At unit scale this parametric model reproduces the reference noise model that Qiskit constructs from the same calibration snapshot to within $3\times10^{-7}$ in average gate infidelity. The rows of Fig.~\ref{fig:noisy_spectra} divide all error probabilities by a common factor and multiply the coherence times by the same factor, so that a single parameter, the median two-qubit error $p_{2q}$, labels each noise level. The demonstration therefore locates the crossover along a single trajectory of jointly improving gate fidelity and coherence; whether the post-selected gadget still leads when the two improve at different rates is left to future work. The zero-noise verification of the pipeline runs the identical chain at a scale factor of $10^{6}$, at which the residual error budget lies more than two orders of magnitude below the shot noise. At this scaling the idle relaxation probabilities are of order $10^{-5}$, so the exact thermal-relaxation channels are within $\sim\!10^{-5}$ of the identity, where their Kraus representation is numerically ill-conditioned; Qiskit's construction breaks trace preservation at the $10^{-8}$ level. The idle channels are therefore represented by their Pauli twirl~\cite{geller_2013,wallman_2016}, a stochastic Pauli channel that matches the exact channel here to a process infidelity below $10^{-5}$.

\section{Compilation and transpilation statistics}\label{app:stats}

Heavy-hex transpilation is stochastic through the routing passes alone: without a coupling map, all reported metrics are bit-identical across the seed pool, which we verified explicitly. With routing, the per-seed spread is substantial; the sequential $M=32$ circuits span two-qubit depths from 3747 to 7990 across 1024 seeds (median 6448), and the extreme draws of depth and of gate count are distinct seeds. Medians and their 95\% confidence intervals are computed by bootstrap with $2\times10^{4}$ resamples. The protocol removes compiler-inserted barriers before transpilation, which gives the transpiler cross-block freedom; retaining them leaves the sequential circuits nearly unchanged but inflates the gadget depth, by 17\% on the all-to-all target.

The round schedule of each Pauli block is a degree-constrained ($f$-)coloring of the interaction graph: the interactions are assigned to rounds so that each spin $k$ occupies at most $n_k = \lceil d_k/R\rceil$ of them. The fewest rounds any such schedule can use is the $f$-chromatic index of Remark~\ref{rem:rho}; we schedule by an earliest-fit greedy, placing each interaction in the earliest round in which both endpoints are below their capacity. For a star this attains the minimum $\max_k \lceil d_k/n_k\rceil$ exactly, realized by distributing the hub couplings round-robin across the register, so the star systems are scheduled optimally on both the gadget and the $R=\Delta$ baseline. On denser graphs the same greedy can exceed the $f$-chromatic index on either side; the tabulated ratios for those systems then compare two greedy schedules rather than two optima, so they are not guaranteed conservative.

The demonstration uses, for each compilation, the transpiler draw of least two-qubit volume $V_{2q}$ over the 1024 seeds. For the sequential compilation this draw is unique; for the gadget, 124 of the 1024 draws share the minimum two-qubit depth and volume, and the tie is broken by the lowest two-qubit gate count $n_{2q}$, then by seed order between two draws differing only in single-qubit layers. The selection is assessed with the first-order per-shot error budget
\begin{equation}
\begin{aligned}
\lambda &= n_{2q}\, e_{2} + w\, T_{\mathrm{circ}}/T_1, \\
T_{\mathrm{circ}} &= d_{2q}\, t_{2q} + (d - d_{2q})\, t_{1q},
\end{aligned}
\end{equation}
with $e_2$ the median two-qubit error, $T_1$ the median relaxation time, $w$ the width of the fixed 16-qubit subgraph, common to both compilations (evaluating the sequential arm on its 13 active qubits narrows but does not close the separation: the distributions remain disjoint), $d$ the total depth, and $t_{2q}$, $t_{1q}$ the layer durations. The same budget certifies the comparison against the routing choice: under the uniform scaling of Appendix~\hyperref[app:noise]{A} both terms of $\lambda$ scale as the inverse scale factor, so its ranking holds at every noise level, and Figure~\ref{fig:lambda} shows that the $\lambda$ distributions of the two compilations over the 1024 draws do not overlap, so every gadget draw carries a smaller error budget than every sequential draw. For this demonstration the volume-optimal selected draws lie at the low-$\lambda$ end of both distributions, the sequential one at its minimum and the gadget one within $6\times10^{-4}$ of its minimum.

\begin{figure}[!htb]
    \centering
    \includegraphics[width=0.801\columnwidth]{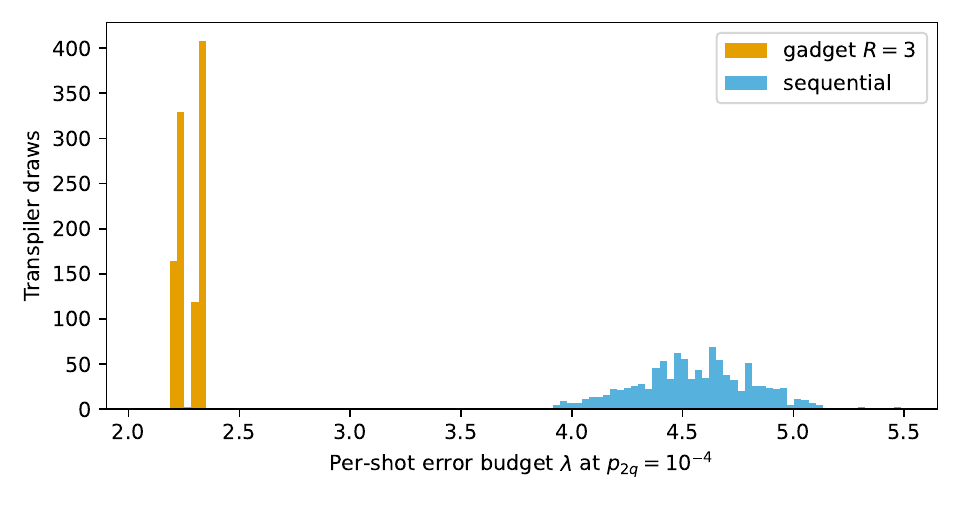}
    \caption{Per-shot error budget $\lambda$ of the two compilations across 1024 transpiler draws on the fixed demonstration subgraph, evaluated at $p_{2q}=10^{-4}$. Each compilation uses its volume-optimal draw, which here also lies at the low-$\lambda$ end of its distribution. The distributions do not overlap: every gadget draw carries a smaller error budget than every sequential draw.}
    \label{fig:lambda}
\end{figure}

\section{Demonstration details}\label{app:demo}

\emph{Sector structure of the star system.} Let $\vec{K} = \sum_{k=1}^{12} \vec{I}_k$ be the total angular momentum of the protons, with quantum number $K \in \{0, \dots, 6\}$, and $\vec{F} = \vec{S} + \vec{K}$ the total spin including the $^{29}$Si spin $\vec{S}$. The Hamiltonian $H = 2\pi J\, \vec{S} \cdot \vec{K}$ commutes with $K^2$ and equals $\pi J [F(F+1) - K(K+1) - \tfrac{3}{4}]$, so every $K \geq 1$ manifold holds exactly two levels, $F = K \pm \tfrac{1}{2}$, split by the frequency $(K + \tfrac{1}{2})J$: the comb; the $K=0$ manifold holds the single level $F = \tfrac{1}{2}$ and contributes no line. The multiplicity of $K$ is $\binom{12}{6-K} - \binom{12}{5-K}$, i.e., 132, 297, 275, 154, 54, 11, 1 for $K = 0, \dots, 6$. The observable commutes with $K^2$ and with the proton permutation symmetry, so the sectors do not mix in the signal. The demonstration's initial basis state (five protons up, seven down) populates the sectors $K = 1, \dots, 6$ with weights 0.375, 0.347, 0.194, 0.068, 0.014, 0.001 and has no $K=0$ component; the weights explain the faint high-$K$ lines of the reference spectra.

\emph{Initial-state dependence.} Figure~\ref{fig:state_dep} compares, on the identical acquisition grid and processing, the spectrum computed from the demonstration's basis state with the spectrum of the standard prepolarized pulse-acquire experiment. The line positions coincide because they are properties of the Hamiltonian; the intensity envelopes differ because they depend on the initial state and detection of each protocol.

\begin{figure}[!htb]
    \centering
    \includegraphics[width=1.0\columnwidth]{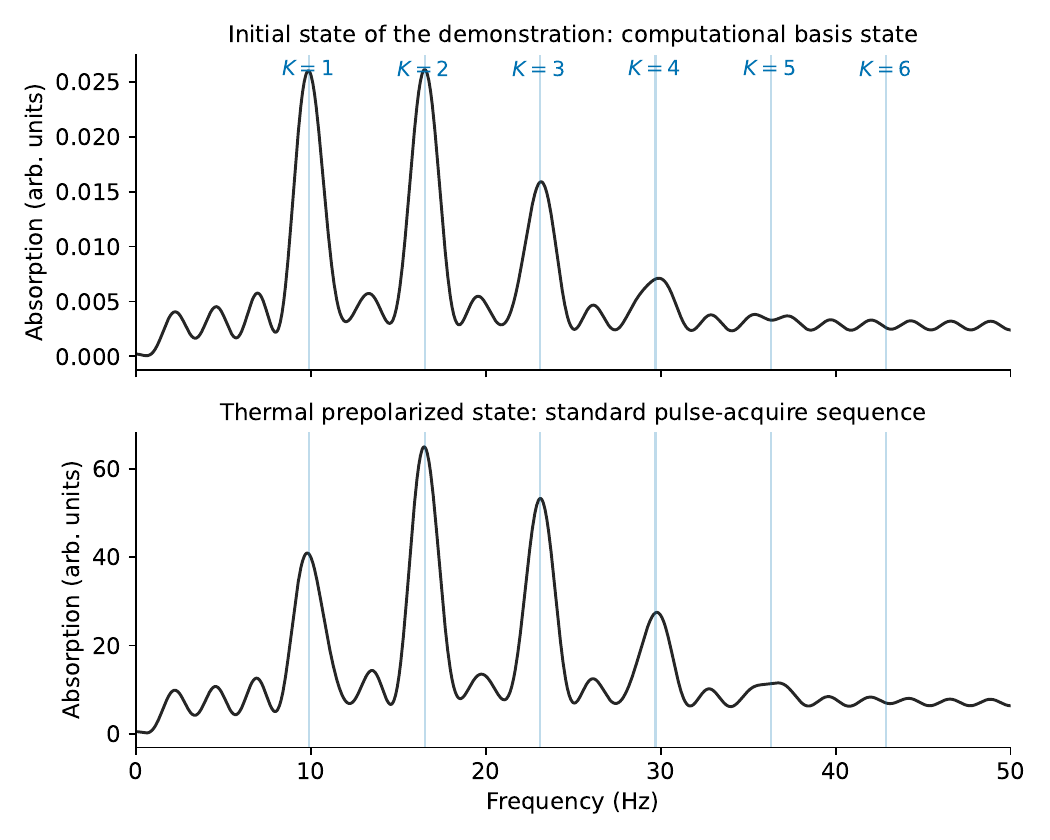}
    \caption{Dependence of the zero-field TMS spectrum on the initial state, computed with SPINACH on the demonstration's 45-point grid and processing. Top: free evolution of the computational basis state of the demonstration. Bottom: the standard prepolarized pulse-acquire experiment (stock zerofield sequence). Vertical lines mark the frequencies $(K+\tfrac{1}{2})J$. Positions and line shapes coincide; the intensity envelopes follow the sector weights of the respective initial states.}
    \label{fig:state_dep}
\end{figure}

\emph{Trotter convergence.} Figure~\ref{fig:nconv} shows the convergence of the second-order product formula on the acquisition grid. The comb-band correlation with the exact spectrum rises from 0.926 at $M=16$ to 0.991 at $M=32$ and 0.998 at $M=48$; $M=32$ is chosen as the point where the agreement saturates at the resolution of the demonstration, and the residual Trotter error is the small difference between the two reference curves of Fig.~\ref{fig:noisy_spectra}.

\begin{figure}[!htb]
    \centering
    \includegraphics[width=1.0\columnwidth]{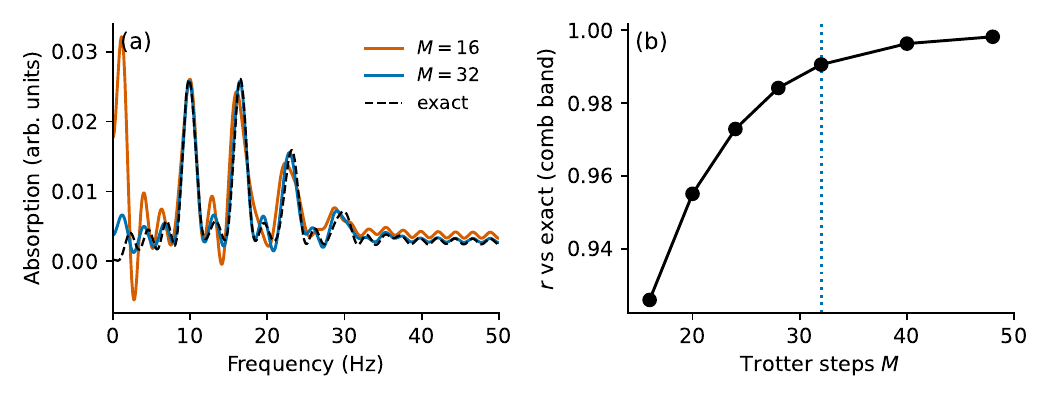}
    \caption{Trotter convergence of the tetramethylsilane (TMS) demonstration reference. (a) Spectra at $M=16$ and $M=32$ against the exact spectrum, on the acquisition grid with 1~Hz line broadening. (b) Comb-band correlation with the exact spectrum as a function of the step count.}
    \label{fig:nconv}
\end{figure}

\section{Star-graph scaling}\label{app:stars}

Figure~\ref{fig:stars} collects the volume-optimal ratios for synthetic star graphs of hub degree $\Delta = 3$ to $60$ on the heavy-hex target, together with the molecules of Table~\ref{tab:molecules}. The stars form the envelope of the attainable gain: at $\Delta = 3$ and $4$ the optimum coincides with the sequential circuit, the gain sets in at $\Delta = 5$, and the optimum improves to 0.30 at $\Delta = 60$. The decrease carries a small even--odd sawtooth: filling the width-minimal schedule takes $\lceil\Delta/2\rceil$ interaction rounds, so an odd hub degree costs one more round than the neighboring even degree against a smoothly growing baseline, leaving odd degrees such as $\Delta = 9$ slightly above the trend. The effect persists in the routing-free limit and across lattices, so it is a property of the discrete schedule rather than transpiler routing or seed noise. The molecules lie on or above the star envelope at their maximum degree; the distance measures how far the interaction graph departs from a single embeddable star: the nearly pure star of HMPA sits close to the envelope, while the dense phosphorus cluster and the low-degree difluoroheptane both sit at parity far above it, the former because its three high-degree hubs, part of a mutually coupled seven-center core, cannot embed together and the latter because it has no hub to fan out.

\begin{figure}[!htb]
    \centering
    \includegraphics[width=0.801\columnwidth]{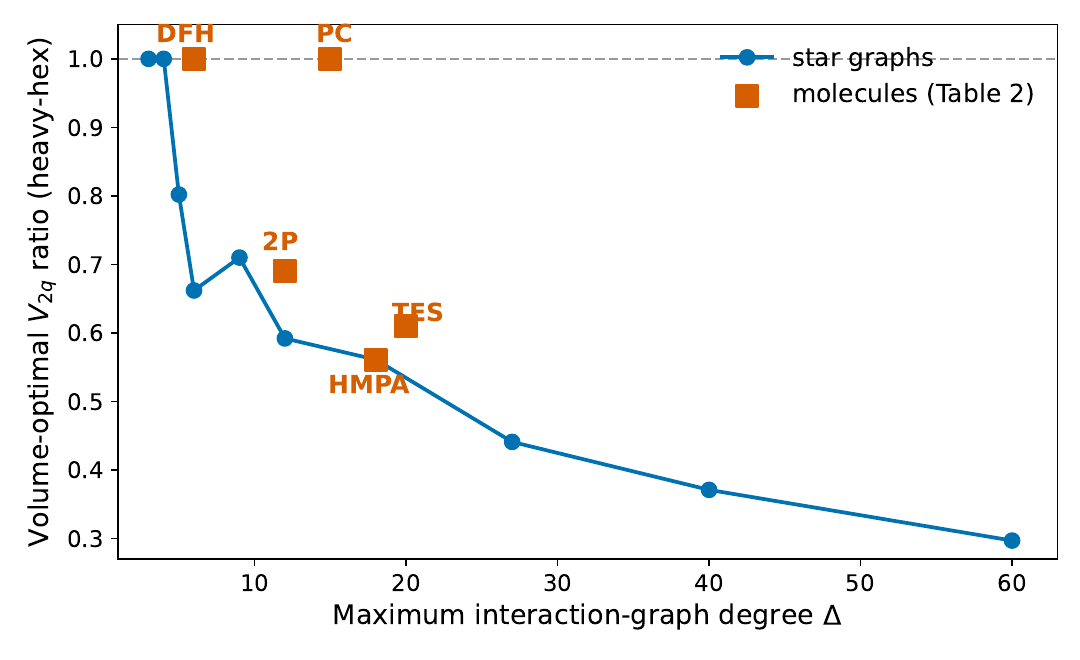}
    \caption{Volume-optimal $V_{2q}$ ratio on the heavy-hex target as a function of the maximum interaction-graph degree $\Delta$: synthetic star graphs (circles, medians over 1024 seeds, or 256 for hub degrees 40 and 60) and the molecules of Table~\ref{tab:molecules} (squares; TMS coincides with the $\Delta=12$ circle and is not drawn separately). The dashed line marks parity with the sequential compilation.}
    \label{fig:stars}
\end{figure}

\section{Device-connectivity dependence}\label{app:connectivity}

Proposition~\ref{prop:staropt} gives the routing-free optimum; on real hardware the routed cost adds to Eq.~\ref{eq:starvol} and depends on how the fan-out embeds in the device graph. The embedding sets a floor: on a device of maximum degree two a copied state reaches at most $2t+1$ qubits after $t$ two-qubit layers, by a light-cone argument, so a logarithmic-depth fan-out to $n$ copies needs coordination above two and otherwise degrades to a linear chain; the routed cost of fan-out is in this sense connectivity-limited~\cite{gokhale2020quantumfanoutcircuitoptimizations}. To map this, we transpiled the volume-optimal schedule for star systems of hub degree $D = 8$, $12$, $18$, and $24$ to devices of increasing coordination number $\kappa$: a line ($\kappa=2$), a heavy-hex lattice, a square grid ($\kappa=4$), a triangular lattice ($\kappa=6$), a king's-graph lattice ($\kappa=8$), and all-to-all connectivity, at 1024 transpiler seeds each. Only the heavy-hex device is a real hardware topology, and it alone is irregular: its qubits are mostly degree two, with degree-three junctions, giving a mean coordination near $2.3$; we plot it at its maximum degree, $\kappa=3$. The other lattices are idealized and included to trace the trend.

The volume-optimal gain is U-shaped in the coordination number (Fig.~\ref{fig:connectivity}). It is weakest at the line ($\kappa=2$), where the fan-out tree cannot branch and the gadget is barely profitable, and weak at all-to-all connectivity for all but the largest hub degree studied, where there is no routing congestion for the fan-out to relieve and the only effect is the added fan-out gates. The optimum sits at an intermediate coordination that matches the fan-out to the device, close to $\kappa \approx D/2$ where this falls within the range tested (hub degrees $8$ and $12$ peak at $\kappa=4$ and $6$); for the higher degrees the peak reaches the densest lattice tested. The gain deepens with hub degree, reaching a volume ratio of $0.19$ for a degree-$24$ hub on the king's graph. The approach to this optimum is not monotonic in $\kappa$: the gadget's routing benefit is concentrated at high coordination (the degree-$24$ fan-out embeds nearly locally only near $\kappa=8$, where its routed depth collapses), whereas the sequential baseline shortens steadily with $\kappa$. At the intermediate triangular point ($\kappa=6$) the baseline therefore gains proportionally more than the gadget and the ratio ticks up ($0.359$ at $\kappa=6$ against $0.337$ at $\kappa=4$) before falling sharply at $\kappa=8$.

\begin{figure}[!htb]
    \centering
    \includegraphics[width=0.872\columnwidth]{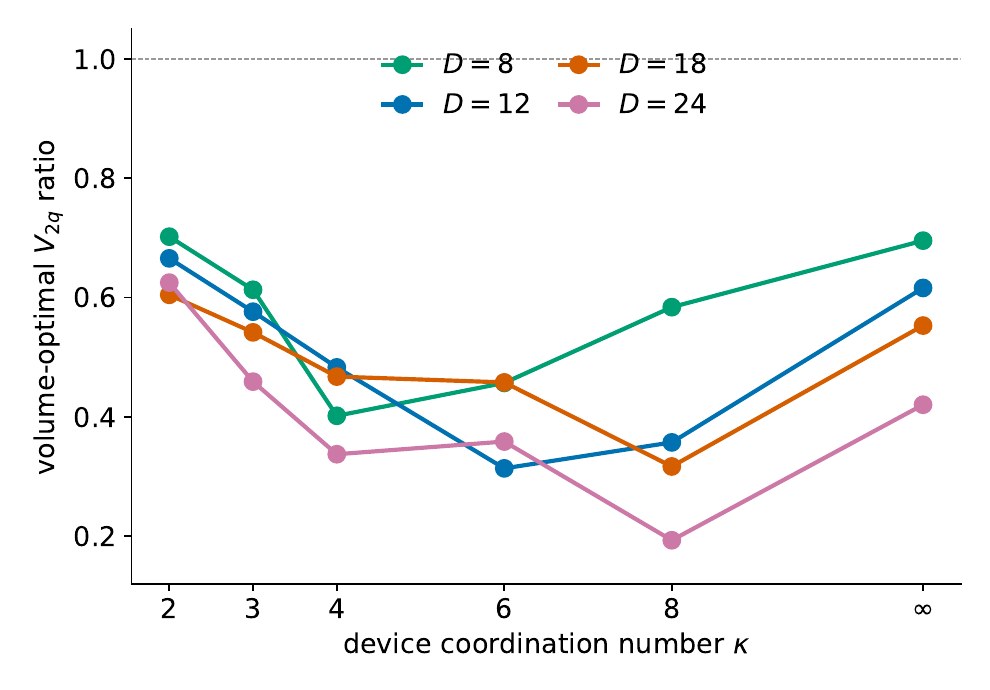}
    \caption{Volume-optimal $V_{2q}$ ratio (relative to the sequential compilation) for star systems of hub degree $D$, as a function of the device coordination number $\kappa$ (medians over 1024 transpiler seeds). The gain is U-shaped, weakest at the line ($\kappa=2$) and, for all but the largest hub degree, at all-to-all connectivity, and strongest at intermediate coordination; the optimum tracks $\kappa \approx D/2$ where it is resolved. Only the heavy-hex point is a real device topology; it is irregular (mean coordination $\approx 2.3$, plotted at its maximum degree $3$), which places its gain between the line and the square grid.}
    \label{fig:connectivity}
\end{figure}

The mechanism is decongestion (Fig.~\ref{fig:decongestion}). A degree-$D$ hub cannot be placed on a qubit of degree $\kappa < D$, so the sequential circuit routes its $D$ interactions through a single congested qubit; the fan-out distributes the hub across a $\lceil D/R\rceil$-qubit register that the device hosts locally, and the total routing distance from the hub to its leaves falls accordingly.

\begin{figure}[!htb]
    \centering
    \includegraphics[width=0.89\columnwidth]{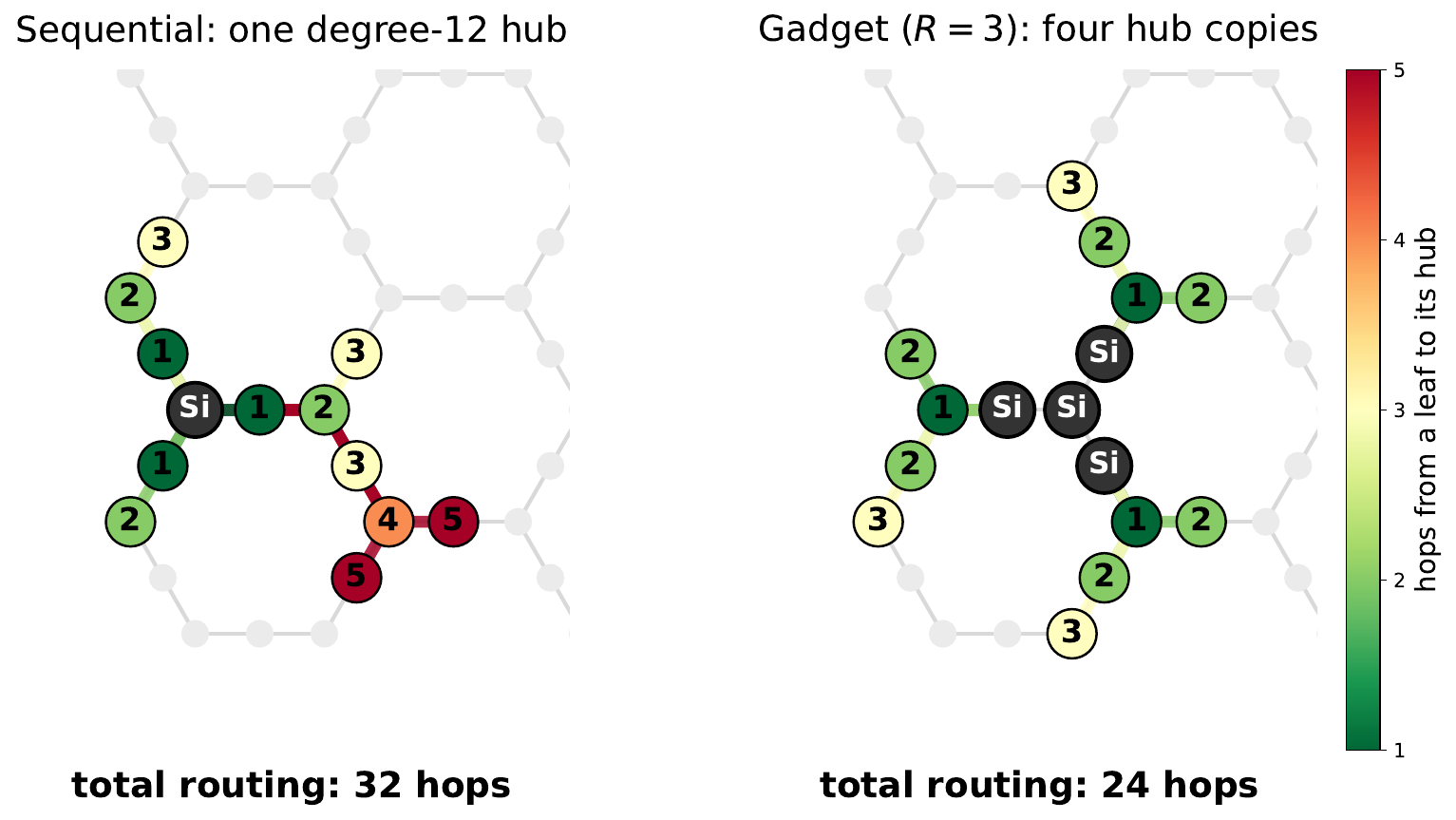}
    \caption{Placement of the 13-spin demonstration star on a heavy-hex lattice (most qubits degree two, with degree-three junctions). Left: the sequential compilation puts the $^{29}$Si hub (dark node) on a single qubit, and its twelve couplings to the proton leaves, colored by hop distance to the hub, must be routed through the congested neighborhood. Right: the gadget at $R=3$ spreads the hub into a four-qubit register that sits among the leaves, shortening the total hub-to-leaf routing distance.}
    \label{fig:decongestion}
\end{figure}

\end{document}